\documentclass[11pt, a4paper]{article}

\usepackage{parskip}
\usepackage{amsmath}
\usepackage{amssymb}
\usepackage{xspace}

\usepackage{graphics}

\usepackage{xcolor}
\definecolor{myblue}{rgb}{0, 0, 0.478}

\usepackage[hyphens]{url}

\usepackage{dsfont}

\usepackage{enumerate}
\usepackage{enumitem}

\usepackage{bm}

\usepackage{algorithm}
\usepackage[noend]{algorithmic}

\newcommand{\ZZ}{\ensuremath{\mathbb{Z}}}
\newcommand{\NN}{\ensuremath{\mathbb{N}}}
\newcommand{\RR}{\ensuremath{\mathbb{R}}}
\newcommand{\QQ}{\ensuremath{\mathbb{Q}}}

\newcommand{\EPS}{\ensuremath{\varepsilon}}
\newcommand{\DELTA}{\ensuremath{\delta}}
\newcommand{\ALPHA}{\ensuremath{\alpha}}
\newcommand{\KAPPA}{\ensuremath{\kappa}}
\newcommand{\MU}{\ensuremath{\mu}}
\newcommand{\GAMMA}{\ensuremath{\gamma}}
\newcommand{\LAMBDA}{\ensuremath{\ell}}
\newcommand{\SIGMA}{\ensuremath{\sigma}}
\newcommand{\DBC}[2]{\ensuremath{\textsc{DBC}_{#1}^{#2}}\xspace}
\newcommand{\PC}[1]{\ensuremath{\textsc{GC}_{#1}}\xspace}
\newcommand{\POPC}[2]{\ensuremath{\textsc{POPC}_{#1}^{#2}}\xspace}
\newcommand{\OPT}{\ensuremath{\textsc{Opt}}\xspace}
\newcommand{\ALG}{\ensuremath{\textsc{Alg}}\xspace}
\newcommand{\PACALG}{\ensuremath{\mathcal{A}}\xspace}
\newcommand{\DNF}{\ensuremath{\textsc{dnf}}\xspace}
\newcommand{\CR}[1]{\ensuremath{\textsc{cr}_{#1}}\xspace}
\newcommand{\EAR}[1]{\ensuremath{\textsc{er}_{#1}^\infty}\xspace}

\newcommand{\ECR}[1]{\ensuremath{\overline{\textsc{cr}}_{#1}}\xspace}
\newcommand{\PROFIT}[1]{\ensuremath{p\left(#1\right)}}
\newcommand{\LS}{\ensuremath{\LAMBDA\textsc{-Splitting}}}
\newcommand{\HY}[3]{\ensuremath{\textsc{Hyb}_{#1}^{#2,#3}}\xspace}
\newcommand{\TL}{\ensuremath{\lambda}\xspace}
\newcommand{\THRESHOLDALG}{\ensuremath{\Phi}\xspace}
\newcommand{\THRESHOLD}[1]{\ensuremath{\Phi_{\hspace{-0.05cm}#1}}\xspace}
\newcommand{\GROUP}[2]{\ensuremath{G_{\hspace{-0.05cm}#1}}^{\hspace{0.03cm}#2}\xspace}
\newcommand{\PLACEHOLDERS}[1]{\ensuremath{P_{\hspace{-0.05cm}#1}}\xspace}

\newcommand{\SIZE}[1]{\ensuremath{\abs{\mathcal{B}_{#1}}}}

\newcommand{\on}[1]{\ensuremath{\operatorname{#1}}\xspace}
\newcommand{\abs}[1]{\ensuremath{\left\lvert #1 \right\rvert}}
\newcommand{\LEV}[1]{\ensuremath{\on{lev}(#1)}}

\usepackage{amsthm}
\newcounter{res_count}

\theoremstyle{plain} 
\newtheorem{theorem}[res_count]{Theorem}

\newtheorem{lemma}[res_count]{Lemma}
\newtheorem{proposition}[res_count]{Proposition}
\newtheorem{observation}[res_count]{Observation}

\theoremstyle{definition}

\title{Online Bin Covering with Frequency Predictions}

\author{
\begin{minipage}{.5\linewidth}
\begin{center}
Magnus Berg\thanks{Berg was supported in part by the Independent Research Fund Denmark, Natural Sciences, grant DFF-0135-00018B and in part by the Innovation Fund Denmark, grant 9142-00001B, Digital Research Centre Denmark, project P40: Online Algorithms with Predictions.} \\ \vspace{0.15cm}
University of Southern Denmark \\
\texttt{magbp@imada.sdu.dk}
\end{center}
\end{minipage}
\begin{minipage}{.49\linewidth}
\begin{center}
Shahin Kamali \\ \vspace{0.15cm}
York University \\
\texttt{kamalis@yorku.ca}
\end{center}
\end{minipage}
}

\date{}

\begin{document}

\maketitle

\begin{abstract}
We study the discrete bin covering problem where a multiset of items from a fixed set $S \subseteq (0,1]$ must be split into disjoint subsets while maximizing the number of subsets whose contents sum to at least $1$.
We study the online discrete variant, where $S$ is finite, and items arrive sequentially. 
In the purely online setting, we show that the competitive ratios of best deterministic (and randomized) algorithms converge to $\frac{1}{2}$ for large $S$, similar to the continuous setting.
Therefore, we consider the problem under the prediction setting, where algorithms may access a vector of frequencies predicting the frequency of items of each size in the instance.
In this setting, we introduce a family of online algorithms that perform near-optimally when the predictions are correct. 
Further, we introduce a second family of more robust algorithms that presents a tradeoff between the performance guarantees when the predictions are perfect and when predictions are adversarial.
Finally, we consider a stochastic setting where items are drawn independently from any fixed but unknown distribution of $S$.
Using results from the PAC-learnability of probabilities in discrete distributions, we also introduce a purely online algorithm whose average-case performance is near-optimal with high probability for all finite sets $S$ and all distributions of $S$.
\end{abstract}

\section{Introduction}

\emph{Bin Covering} is a classical NP-complete~\cite{AJKL84} optimization problem where the input is a multiset of items, each with a size between $0$ and $1$.
The objective is to split the items into disjoint subsets, called \emph{bins}, while maximizing the number of bins whose contents sum to at least $1$~\cite{JS03}. The problem is often considered a dual to the bin packing problem, which asks for minimizing the number of bins, subject to each bin having a sum of at most 1.

In the online covering problem~\cite{CT87,CFLZ99,AJKL84}, items arrive one by one, and whenever an item arrives, an algorithm has to irrevocably place this item in an existing bin or open a new bin to place the item in. 
The existing results mostly consider a continuous setting in which items take any real value from $(0,1]$, and it is well known that a simple greedy strategy, \emph{Dual-Next-Fit} ($\DNF$), achieves an optimal competitive ratio of $\frac{1}{2}$~\cite{AJKL84}.

In this paper, we consider a discrete variant of Online Bin Covering, where item sizes belong to a finite, known set $S \subseteq (0,1]$. 
We abbreviate this problem by $\DBC{S}{}$.
The special case when $S = \{\frac{i}{k} \mid i=1,\ldots,k\}$ has been studied in the previous work. For example, Csirik, Johnson, and Kenyon~\cite{CJK01} developed online algorithms with good average-case performance based on the \emph{Sum of Squares} algorithm for Online Discrete Bin Packing~\cite{CJKSW99,CJKOSW06}. 
In this paper, we consider a more general setting only where $S$ may be \emph{any} finite subset of $(0,1]$.


For measuring and comparing the quality of online algorithms for the $\DBC{S}{}$ problem, we rely on the classical \emph{competitive analysis} framework~\cite{BE98,K16}, where one measures the quality of an online algorithm by comparing the performance of the algorithm to the performance of an optimal offline algorithm optimizing for the best worst-case guarantee.

\subsection{Previous Work}

The possibilities for creating algorithms for Online Bin Covering are well-studied.
In the continuous setting, where items can take any size in $(0,1]$, Assmann et al.~\cite{AJKL84} present the $\frac{1}{2}$-competitive algorithm, $\DNF$, and Csirik and Totik~\cite{CT87} present an impossibility result showing that no deterministic algorithm can achieve a better competitive ratio.
Later,  Epstein~\cite{E01} proved that the same impossibility result holds for randomized algorithms as well.
Online Bin Covering has been studied under the advice setting~\cite{BFKL21,BNV23}, where algorithms can access an advice tape that has encoded information about the input sequence. 
The aim is to determine how much additional information, measured by the number of bits needed to encode the information, is necessary and sufficient to create online algorithms that are better than $\frac{1}{2}$-competitive and how well can algorithms perform when they are given a certain amount of information. 

In recent years, developments in machine learning have inspired the question of how online algorithms may benefit from machine-learned advice~\cite{LV21,KPS18}, commonly referred to as \emph{predictions}.
Unlike the advice model, the predictions may be erroneous or even adversarial.
Online algorithms with predictions is a rapidly growing field (see, e.g.,~\cite{ALPS}) that aims at deriving online algorithms that provide a tradeoff between \emph{consistency} and \emph{robustness}. The consistency of an online algorithm with prediction refers to its competitive ratio when predictions are error-free; ideally, the consistency of an algorithm is $1$ or close to 1. On the other hand, robustness refers to the competitive ratio assuming adversarial predictions; ideally, the robustness of an algorithm is close to the competitive ratio of the best purely online algorithm (with no prediction).
This ideal case is often not realizable, and so one often accompanies an analysis of an online problem with predictions by a consistency/robustness trade-off~\cite{ADJKR20,KPS18,WZ20}, giving explicit bounds on an algorithm's consistency as a function of its robustness, and vice versa.

To the authors' knowledge, there does not exist any previous work on Bin Covering with predictions. 
There do, however, exist previous work on the related bin packing problem~\cite{AKS22,ADJKR20}.

\subsection{Contribution}
Our contributions for $\DBC{S}{}$ can be summarized as follows. 
Throughout, we let $k = |S|$. Our results concern the power of prediction and learning in improving online discrete bin covering algorithms.
In the continuous setting, where items take \emph{any} real value in $(0,1]$, it is established that no improvements in the competitive ratio can be achieved via predictions that are of size independent of input length, even if the predictions are error-free. This follows from a result of~\cite{BFKL21} that states any algorithm with an advice of size $o(\log \log n)$ is no better than 2-competitive. This negative result yields relaxing the setting such that items come from a fixed, finite set, which is also studied in the related bin packing problem~\cite{AKS22}. 

\begin{itemize}
\item[] \textbf{Purely online setting:} 
We briefly state some results on purely online algorithms for $\DBC{F_k}{}$, where $F_k = \{\frac{i}{k} \mid i = 1,2,\ldots,k\}$.
Based on ideas from~\cite{CT87} and~\cite{E01}, we show that:

\begin{theorem}\label{thm:impossiblity_dbc_k}
Let $\ALG$ be any deterministic or randomized online algorithm for $\DBC{F_k}{}$, with $k \geqslant 5$. 
Then, the competitive ratio of $\ALG$ is at most $\frac{1}{2} + \frac{1}{H_{k-1}}$, where $H_{k-1} = \sum_{i=1}^{k-1} \frac{1}{i}$.
\end{theorem}

The proof for the deterministic case can be found in Theorem~\ref{thm:impossibility_deterministic_app}, and using Yao's Principle~\cite{Y77,BE98,K16} we generalize to the randomized case in Theorem~\ref{thm:impossibility_randomized_app} (see Appendix~\ref{sec:online_setting}).

Two immediate consequences of Theorem~\ref{thm:impossiblity_dbc_k} are the well-known facts~\cite{CT87,E01}, that the competitive ratio of any deterministic or randomized algorithm for Online Bin Covering is at most $\frac{1}{2}$. 
This shows that Online Bin Covering is still a hard problem, even after discretization.

\item[] \textbf{Prediction setting:} 
We study the $\DBC{S}{}$ problem under a prediction setting, where predictions concerning the frequency of item sizes are available. 
We start with an impossibility result that establishes a consistency/robustness tradeoff for this prediction scheme (Theorem~\ref{thm:consistency_robustness_tradeoff}).
We then present a family of online algorithm, named \emph{Group Covering}, which are near-optimal when the predictions are error-free, for all finite sets $S \subseteq (0,1]$ (Theorem~\ref{thm:consistency_of_profilecovering}). 
Further, we create a family of hybrid algorithms that accepts a parameter $\TL$, quantifying one's trust in the predictions. 
We establish a consistency/robustness tradeoff that bounds the consistency and robustness of these hybrid algorithms as a function of $\TL$ (Theorems~\ref{thm:consistency_hybrid} and~\ref{thm:robustness_hybrid}).
\item[] \textbf{Stochastic setting:} 
Motivated by the work of Csirik, Johnson, and Kenyon~\cite{CJK01}, we study the problem under a stochastic setting, where item sizes follow an unknown distribution. 
Unlike~\cite{CJK01}, which assumes items are of sizes $\frac{i}{k}$, for $i = 1,2,\ldots,k$, we do not make any assumption about input set $S$.  
We use a PAC-learning bound~\cite{C20,SB14} to create a family of online algorithms without predictions, whose expected ratio~\cite{CJK01} is near-optimal with high probability, for any finite set $S$, and any distribution $D$ of $S$ (Theorem~\ref{thm:expected_ratio_popc}).
\end{itemize}

\section{Preliminaries}

\subsection{Online Discrete Bin Covering}
Fix a finite set $S = \{s_1,s_2,\ldots,s_k\} \subseteq (0,1]$. 
An instance for \emph{$S$-Discrete Bin Covering} is a sequence $\SIGMA = \langle a_1,a_2,\ldots,a_n\rangle$ of items, where $a_i \in S$, for $i\in[n]$. 
The task of an algorithm $\ALG$ is to place the items in $\SIGMA$ into bins $B_1,B_2,\ldots,B_t$, maximizing the number of bins, $B$, for which $\sum_{a \in B} a \geqslant 1$. 
For any bin, $B$, we call $\LEV{B} = \sum_{a' \in B} a'$ the \emph{level} of $B$.
Unless otherwise mentioned, we assume that algorithms are aware of $S$.
In the online setting, the items are presented one-by-one to $\ALG$, and upon receiving an item $a$, $\ALG$ has to place $a$ in a bin.
This decision is irrevocable. 
We abbreviate \emph{Online $S$-Discrete Bin Covering} by $\DBC{S}{}$.

Throughout, we assume that $k \geqslant 2$.
Further, for any $k \in \ZZ^+$, we set $F_k = \left\{ \frac{i}{k} \mid \mbox{for $i=1,2,\ldots,k$}\right\}$, and abbreviate $\DBC{F_k}{}$ by $\DBC{k}{}$.

\subsection{Performance Measures}
Given an online maximization problem, $\Pi$, an online algorithm, $\ALG$, for $\Pi$, and an instance, $\SIGMA$, of $\Pi$, we let $\ALG[\SIGMA]$ be $\ALG$'s solution on instance $\SIGMA$ and $\ALG(\SIGMA)$ be the profit of $\ALG[\SIGMA]$.
If $\ALG$ is deterministic, then the \emph{competitive ratio} of $\ALG$ is
\begin{align*}
\CR{\ALG} = \sup\{c \in (0,1] \mid \exists b > 0 \colon \forall \SIGMA \colon \ALG(\SIGMA) \geqslant c \cdot \OPT(\SIGMA) - b \},
\end{align*}
where $\OPT$ is an offline optimal algorithm for $\Pi$.
Further, $\ALG$ is said to be \emph{$c$-competitive} if $c \leqslant \CR{\ALG}$.

For a fixed finite set $S = \{s_1,s_2,\ldots,s_k\} \subseteq (0,1]$, and a fixed (unknown) distribution $D$ of $S$, the \emph{asymptotic expected ratio}~\cite{CW98,CJK01} of an online algorithm, $\ALG$, is
\begin{align}\label{eq:asymptotic_expected_ratio}
\EAR{\ALG}(D) = \liminf_{n\to\infty} \mathbb{E}_D \left[ \frac{\ALG(\SIGMA_n(D))}{\OPT(\SIGMA_n(D))} \right],
\end{align}
where $\SIGMA_n(D)$ is a sequence of $n$ independent identically distributed random variables, $\SIGMA_n(D) = \langle X_1,X_2,\ldots,X_n\rangle$, where $X_i \sim D$, for all $i =1,2,\ldots,n$\footnote{The particular choice of notation is due to the items being random variables rather than fixed items.}.


When an algorithm, $\ALG$, has access to predictions, the \emph{consistency} of $\ALG$, and the \emph{robustness} of $\ALG$, is $\ALG$'s competitive ratio when the predictions are error-free and adversarial, respectively.

Throughout, we let $[n] = \{1,2,\ldots,n\}$.

\section{Predictions Setting}

In this section, we assume that algorithms are given a \emph{frequency prediction}, which, for a fixed instance $\SIGMA$, and each item $s_i \in S$, predicts what fraction of items in $\SIGMA$ are of size $s_i$. 
We analyze how algorithms may benefit from this additional information, both when the information is reliable and when it is unreliable.

Formally, given a finite set $S = \{s_1,s_2,\ldots,s_k\} \subseteq (0,1]$, and an instance, $\SIGMA$, of $\DBC{S}{}$, we let $n_i^\SIGMA$ be the number of items of size $s_i$ in $\SIGMA$, $n^\SIGMA$ be the total number of items in $\SIGMA$, and $f_i^\SIGMA = \frac{n_i^\SIGMA}{n^\SIGMA}$. 
We call $f_i^\SIGMA$ the \emph{frequency} of items of size $s_i$ in $\SIGMA$, and set $\bm{{f^\SIGMA}} = (f_1^\SIGMA,f_2^\SIGMA,\ldots,f_k^\SIGMA)$.
When $\SIGMA$ is clear from the context, we abbreviate $n_i^\SIGMA$, $n^\SIGMA$, $f_i^\SIGMA$, and $\bm{{f^\SIGMA}}$, by $n_i$, $n$, $f_i$, and $\bm{{f}}$, respectively.
Further, we sometimes write $n_i^N$ for $n_i^\SIGMA$, when $N$ is a solution to $\SIGMA$.

Throughout, we abbreviate $\DBC{S}{}$ with frequency predictions by $\DBC{S}{\mathcal{F}}$. 
An instance for $\DBC{S}{\mathcal{F}}$ is a tuple $(\SIGMA,\bm{{\hat{f}}})$ consisting of a sequence of items, $\SIGMA$, and a vector of predicted frequencies $\bm{{\hat{f}}} = \left( \hat{f}_1, \hat{f}_2,\ldots,\hat{f}_k \right)$.

It is well-known that probabilities in discrete distributions are PAC-learnable, as shown in~\cite{C20}. 
That is, there exists a polynomial-time algorithm that learns the probabilities in discrete distributions to arbitrary precision with a confidence that is arbitrarily close to 1, given sufficiently many random samples (see~\cite{SB14} for a formal definition of PAC-learnability).
This makes frequency predictions easily attainable when a lot of historical data is available.

\subsection{A Consistency-Robustness Trade-Off for $\DBC{\mathcal{S}_k}{\mathcal{F}}$}

In the following, by a \emph{wasteful} algorithm, we mean an algorithm that sometimes places an item, $I$, in a bin, $B$, for which $\LEV{B} \geqslant 1$ before $I$ was placed in $B$.
Any wasteful algorithm can be trivially converted to an algorithm that is equally good (possibly better) and avoids placing items into already-covered bins.
Therefore, in what follows, we assume that all algorithms are non-wasteful. 


\begin{theorem}\label{thm:consistency_robustness_tradeoff}
Any $(1-\ALPHA)$-consistent deterministic algorithm for $\DBC{k}{\mathcal{F}}$ is at most $2\ALPHA$-robust.
\end{theorem}
\begin{proof}
Let $\ALG$ be any deterministic online algorithm for $\DBC{\mathcal{S}_k}{\mathcal{F}}$.
Consider the instance $(\SIGMA_1^n,\bm{{\hat{f}}})$, with $\bm{{\hat{f}}} = ( \hat{f}_{s_1},\hat{f}_{s_2},\ldots,\hat{f}_{s_k} )$, where
\begin{align*}
\SIGMA_1^n = \left\langle \left\langle \frac{k-1}{k} \right\rangle^n , \left\langle \frac{1}{k} \right\rangle^n \right\rangle \hspace{0.2cm} \text{and} \hspace{0.2cm} \hat{f}_{s_i} = \begin{cases}
\frac{1}{2}, &\mbox{if $s_i \in \{\frac{1}{k},\frac{k-1}{k}\}$} \\
0, &\mbox{otherwise.}
\end{cases}
\end{align*}
Clearly, $\bm{{\hat{f}}}$ is a perfect prediction for $\SIGMA_1^n$, and $\OPT(\SIGMA_1^n) = n$.
Hence, by the consistency of $\ALG$, there exists a constant $b$, such that
\begin{align}\label{eq:alg_consistency}
\ALG(\SIGMA_1^n,\bm{{\hat{f}}}) \geqslant (1-\ALPHA) \cdot \OPT(\SIGMA_1^n) - b = (1-\ALPHA) \cdot n - b.
\end{align}
Let $\mathcal{B}_i$, for $i = 1,2$, be the collection of bins that $\ALG$ places $i$ items of size $\frac{k-1}{k}$ in.
Then,
\begin{align*}
\ALG(\SIGMA_1^n , \bm{{\hat{f}}}) \leqslant \SIZE{1} + \SIZE{2} + \frac{n - \SIZE{1}}{k}.
\end{align*}
Since $\ALG$ is non-wasteful, $n = \SIZE{1} + 2\cdot\SIZE{2}$, and so, by Equation~\eqref{eq:alg_consistency}, 
\begin{align*}
(1-\ALPHA) \cdot (\SIZE{1} + 2\cdot\SIZE{2}) - b \leqslant \SIZE{1} + \frac{(k+2) \cdot \SIZE{2}}{k}.
\end{align*}
This implies a lower bound on $\SIZE{1}$ of
\begin{align}\label{eq:lower_bound_on_size_1}
\frac{n\cdot \left( 1 - 2\cdot \ALPHA - \frac{2}{k}\right) - 2\cdot b}{1 - \frac{2}{k}} \leqslant \SIZE{1}.
\end{align}
Hence, to be $(1-\ALPHA)$-consistent then, after the first $n$ requests from $(\SIGMA_1^n,\bm{{\hat{f}}})$, $\ALG$ must have at least $\frac{n\cdot \left( 1 - 2\cdot \ALPHA - \frac{2}{k}\right) - 2\cdot b}{1 - \frac{2}{k}}$ open bins with exactly one item of size $\frac{k-1}{k}$ in.

Consider the instance $(\SIGMA_2^n,\bm{{\hat{f}}})$, with imperfect predictions, where 
\begin{align*}
\SIGMA_2^n = \left\langle \frac{k-1}{k} \right\rangle^n.
\end{align*}
Since the first $n$ requests of $\SIGMA_1^n$ and $\SIGMA_2^n$ are identical, $\ALG$ cannot distinguish the instances $(\SIGMA_1^n ,\bm{{\hat{f}}})$ and $(\SIGMA_2^n,\bm{{\hat{f}}})$ until it has seen the first $n$ items.
Hence, since $\ALG$ is deterministic, it distributes the first $n$ items identically on the two instances.
Since $n = \SIZE{1} + 2\cdot\SIZE{2}$, Equation~\eqref{eq:lower_bound_on_size_1} implies that
\begin{align*}
\ALG(\SIGMA_2^n, \bm{{\hat{f}}}) \leqslant \SIZE{2} = \frac{n-\SIZE{1}}{2} \leqslant \frac{n - \frac{n\cdot \left( 1 - 2\cdot \ALPHA - \frac{2}{k}\right) - 2\cdot b}{1 - \frac{2}{k}}}{2} = \frac{2 \cdot n \cdot \ALPHA + 2 \cdot b}{2 - \frac{4}{k}}.
\end{align*}
Since $\OPT(\SIGMA_2^n) = \frac{n}{2}$, then, for all $n \in \ZZ^+$,
\begin{align*}
\frac{\ALG(\SIGMA_2^n,\bm{{\hat{f}}})}{\OPT(\SIGMA_2^n)} \leqslant \frac{\frac{2 \cdot n \cdot \ALPHA + 2 \cdot b}{2 - \frac{4}{k}}}{\frac{n}{2}} = \frac{4 \cdot n \cdot \ALPHA + 4 \cdot b}{n \cdot \left( 2 - \frac{4}{k}\right)} \leqslant 2\cdot \ALPHA - \frac{2\cdot b}{n},
\end{align*}	
and so $\ALG$ is at most $2\cdot \ALPHA$-robust.
\end{proof}

Observe, by the above proof, that Theorem~\ref{thm:consistency_robustness_tradeoff} applies to a wider range of sets $S$, beyond $S = F_k$. 
In particular, we only use two items in the proof, implying that Theorem~\ref{thm:consistency_robustness_tradeoff} can be stated for all finite sets $S \subseteq (0,1]$, for which $\{\frac{1}{k}, \frac{k-1}{k}\} \subseteq S$.

\subsection{A Near-Optimally Consistent Algorithm for $\DBC{S}{\mathcal{F}}$}

In this section, inspired by the \emph{Profile Packing} algorithm from~\cite{AKS22}, we present a family of algorithms named \emph{Group Covering}, parameterized by a parameter, $\EPS$, that receives frequency predictions, and outputs a $(1-\EPS$)-approximation of the optimal solution, assuming predictions are error-free. 
In other words, the algorithm achieves a consistency that is arbitrarily close to optimal.
For a fixed $\EPS > 0$, we let $\PC{\EPS}$ be the Group Covering algorithm given $\EPS$.

\paragraph{The Strategy of Group Covering}

Fix a finite set $S = \{s_1,s_2,\ldots,s_k\} \subseteq (0,1]$.
A \emph{non-wasteful bin type} is an ordered $l$-tuple $(a_1,a_2,\ldots,a_l)$ of items, with $a_i \in S$, for all $i \in [l]$, such that $a_1$ was placed in the bin first, then $a_2$, and so on, and such that $\sum_{i=1}^{l-1} a_i < 1$.
Observe that our definition implies an ordering of the items in bin types.
This ordering is essential for well-definedness in the case of Bin Covering. 
For example, we consider the bin type $(1/2,1/2,\EPS)$ wasteful, as the bin is covered after placing the second item of size $1/2$, but the bin type $(1/2,\EPS,1/2)$ is non-wasteful, as removing the top item will make the bin no longer covered.
Observe that, according to our definition, non-covered bins are also filled with respect to a non-wasteful bin type.
We let $\mathcal{T}_S$ denote the collection of all possible non-wasteful bin types given $S$, and set $\tau_S = \abs{\mathcal{T}_S}$ and $\tau_S^{m} = \on{max}_{t \in \mathcal{T}_S} \{\abs{t}\}$.  
For example, if $S = \left\{\frac{1}{k},\frac{k-1}{k}\right\}$ then,

\scalebox{.9}{
\begin{minipage}{1.1\textwidth}
\begin{align*}
    \mathcal{T}_S &= \left\{ \left\{\frac{k-1}{k}\right\} , \left\{ \frac{k-1}{k} , \frac{1}{k}\right\} , \left\{ \frac{k-1}{k} , \frac{k-1}{k}\right\}\right\} \\
    &\cup \left\{ \left\{ \underbrace{\frac{1}{k} , \frac{1}{k} , \ldots, \frac{1}{k}}_{\mbox{$i$ times}}  \right\} \mid i=1,2,\ldots,k \right\} \\
    &\cup \left\{ \left\{ \underbrace{\frac{1}{k} , \frac{1}{k} , \ldots, \frac{1}{k}}_{\mbox{$i$ times}} , \frac{k-1}{k}  \right\} \mid i=1,2,\ldots,k-1 \right\},
\end{align*}
\end{minipage}
}

$\tau_S = 2k + 2$ and $\tau_S^m = k$.

By the arguments in the previous subsection, there exists a non-wasteful optimal algorithm.
Throughout, we assume that $\OPT$ is non-wasteful.

For a fixed set, $S = \{s_1,s_2,\ldots,s_k\} \subseteq (0,1]$, and any $\EPS > 0$, the strategy of $\PC{\EPS}$ is given in Algorithm~\ref{alg:profile_covering}.
In words, given an instance of $\DBC{S}{\mathcal{F}}$, $(\SIGMA,\bm{{\hat{f}}})$, $\PC{\EPS}$ creates an optimal solution to the following sequence, based on $\bm{{\hat{f}}}$, 
\begin{align*}
\SIGMA_{\on{sub}} = \langle \lfloor \hat{f}_1 \cdot m_{k,\EPS} \rfloor , \lfloor \hat{f}_2 \cdot m_{k,\EPS} \rfloor, \ldots, \lfloor \hat{f}_k \cdot m_{k,\EPS} \rfloor \rangle,
\end{align*}
where $m_{k,\EPS} = m_{\EPS} + k$, and $m_{\EPS} = \lceil 3 \cdot \tau_S \cdot \tau_S^m \cdot \EPS^{-1} \rceil$, except that all items, $a$, in $\OPT[\SIGMA_{\on{sub}}]$ has been replaced with placeholders of size $a$.
A \emph{placeholder of size $a$} is a virtual item of size $a$, which reserves space for an item of size $a$.
We let $\PLACEHOLDERS{\bm{{\hat{f}}},\EPS}$ be the copy of $\OPT[\SIGMA_{\on{sub}}]$ containing placeholders.
To finish the initialization, $\PC{\EPS}$ opens the first \emph{group}, $\GROUP{\bm{{\hat{f}}},\EPS}{1}$; a copy of $\PLACEHOLDERS{\bm{{\hat{f}}},\EPS}$.

When an item, $a$, arrives, $\PC{\EPS}$ searches for a placeholder of size $a$ in the open groups, always searching in $\GROUP{\bm{{\hat{f}}},\EPS}{1}$ first, then $\GROUP{\bm{{\hat{f}}},\EPS}{2}$ second, and so on.
If such a placeholder exists, $\PC{\EPS}$ replaces the placeholder with $a$.
If no such placeholder exists, $\PC{\EPS}$ checks whether $\PLACEHOLDERS{\bm{{\hat{f}}},\EPS}$ contains such a placeholder, by checking whether $a \in \SIGMA_{\on{sub}}$.
If so, then $\PC{\EPS}$ opens a new group, $\GROUP{\bm{{\hat{f}}},\EPS}{i}$, i.e.\ a new copy of $\PLACEHOLDERS{\bm{{\hat{f}}},\EPS}$.
Then, it replaces a newly created placeholder with $a$.
Otherwise, then $\PC{\EPS}$ places $a$ in an \emph{extra}-bin using $\DNF$.
Extra bins are special bins reserved for items that $\PC{\EPS}$ did not expect to receive any of.

\paragraph{Analysis of $\PC{\EPS}$}

We say that a group, $\GROUP{\bm{{\hat{f}}},\EPS}{i}$, is \emph{completed} if all placeholders have been replaced by items, and we let $g_{\EPS}$ be the number groups that $\PC{\EPS}$ completes.
By construction, $\PC{\EPS}$ first completes $\GROUP{\bm{{\hat{f}}},\EPS}{1}$, then $\GROUP{\bm{{\hat{f}}},\EPS}{2}$, and so on.

\begin{lemma}\label{lem:num_completed_groups}
Fix any finite set $S = \{s_1,s_2,\ldots,s_k\} \subseteq (0,1]$, any $\EPS \in (0,1)$, and any instance $(\SIGMA,\bm{{\hat{f}}})$ for $\DBC{S}{\mathcal{F}}$, with $\bm{{\hat{f}}} = \bm{{f}}$.
Then, $\left\lfloor \frac{n}{m_{k,\EPS}} \right\rfloor \leqslant g_{\EPS} \leqslant \left\lfloor \frac{n}{m_{\EPS}} \right\rfloor$. 
\end{lemma}
\begin{proof}
Since $\sum_{i=1}^k \lfloor f_i \cdot m_{k,\EPS} \rfloor \geqslant m_{\EPS}$, no group contains less than $m_{\EPS}$ items, and so $g_{\EPS} \leqslant \left\lfloor \frac{n}{m_{\EPS}} \right\rfloor$.
Further, for all $i \in [k]$, $\left\lfloor \frac{n}{m_{k,\EPS}} \right\rfloor \cdot \lfloor f_i \cdot m_{k,\EPS} \rfloor \leqslant n_i$, and so all placeholders of size $s_i$ in at least $\left\lfloor \frac{n}{m_{k,\EPS}} \right\rfloor$ groups are replaced by an item. 
\end{proof}

\begin{algorithm}
\caption{$\PC{\EPS}$}
\begin{algorithmic}[1]
\STATE \textbf{Input:} a $\DBC{S}{\mathcal{F}}$-instance. $(\SIGMA,\bm{{\hat{f}}})$
\STATE $j,l \leftarrow 1$ \label{line:pc_initialization_start}
\STATE Compute $\tau_S$, $\tau_S^{m}$, and $k = \abs{S}$ 
\STATE $m_{\EPS} \leftarrow  \lceil 3 \cdot \tau_S \cdot \tau_S^{m} \cdot \EPS^{-1} \rceil$
\STATE $m_{k,\EPS} \leftarrow m_{\EPS} + k$
\STATE $\SIGMA_{\on{sub}} \leftarrow \langle \lfloor \hat{f}_1 \cdot m_{k,\EPS} \rfloor, \lfloor \hat{f}_2 \cdot m_{k,\EPS} \rfloor,\ldots, \lfloor \hat{f}_k \cdot m_{k,\EPS} \rfloor \rangle$ 
\STATE $\PLACEHOLDERS{\bm{{\hat{f}}},\EPS} \leftarrow \emptyset$
\FORALL {$B \in \OPT[\SIGMA_{\on{sub}}]$}
	\STATE $B' \leftarrow \emptyset$ \COMMENT{Create a new empty bin}
	\FORALL {$a \in B$}
		\STATE $B' \leftarrow B' \cup \{p_a\}$ \COMMENT{Add a placeholder of size $a$ to $B'$}
	\ENDFOR
	\STATE $\PLACEHOLDERS{\bm{{\hat{f}}},\EPS} \leftarrow \PLACEHOLDERS{\bm{{\hat{f}}},\EPS} \cup B'$ \COMMENT{Add a copy of $B$ containing placeholders to $\PLACEHOLDERS{\bm{{\hat{f}}},\EPS}$} 
\ENDFOR
\STATE $\GROUP{\bm{{\hat{f}}},\EPS}{1} \leftarrow \PLACEHOLDERS{\bm{{\hat{f}}},\EPS}$ \COMMENT{Open the first group} \label{line:pc_initialization_end}
\WHILE {receiving items, $a$,} \label{line:pc_distribution_start}
	\STATE $b \leftarrow \texttt{true}$ \COMMENT{Marks whether $a$ still has to be placed} 
	\FOR[Go through open groups chronologically] {$i = 1,2,\ldots,l$} 
		\IF[To avoid trying to place $a$ multiple times] {$b$}		
			\IF[Search for $p_a$ in $\GROUP{\bm{{\hat{f}}},\EPS}{i}$] {$\exists B \in \GROUP{\bm{{\hat{f}}},\EPS}{i} \colon p_a \in B$}
				\STATE $B \leftarrow B \setminus \{p_a\} \cup \{a\}$ \COMMENT{Swap out placeholder, $p_a$, for $a$}
				\STATE $b \leftarrow \texttt{false}$ \COMMENT{$a$ has been placed in a bin} 
			\ENDIF
		\ENDIF
	\ENDFOR
	\IF[Checking whether $a$ has been placed] {$b$}
		\IF[Checking whether $a \in \SIGMA_{\on{sub}}$] {$\lfloor \hat{f}_{a} \cdot m_{k,\EPS} \rfloor \neq 0$}
			\STATE $l \leftarrow l + 1$
			\STATE $\GROUP{\bm{{\hat{f}}},\EPS}{l} \leftarrow \OPT[\SIGMA_{\on{sub}}]$ \COMMENT{Open a new group}
			\STATE Determine $B \in \GROUP{\bm{{\hat{f}}},\EPS}{l}$ such that $p_a \in B$, and $B \leftarrow B\setminus\{p_a\} \cup \{a\}$
		\ELSE[$a \not\in \SIGMA_{\on{sub}}$]
			\STATE $B^R_j \leftarrow B^R_j \cup \{a\}$ \COMMENT{Place $a$ in a \emph{random} bin using $\DNF$}
			\IF {$\LEV{B^R_j} \geqslant 1$}
				\STATE $j \leftarrow j+1$
				\STATE $B^R_j \leftarrow \emptyset$ \label{line:pc_distribution_end}
			\ENDIF
		\ENDIF
	\ENDIF
\ENDWHILE
\end{algorithmic}
\label{alg:profile_covering}
\end{algorithm}

Throughout, given a solution, $N$, to a collection of items, $\SIGMA$, we let $\PROFIT{N}$ be the profit of $N$.
Observe that $\PROFIT{\GROUP{\bm{{\hat{f}}},\EPS}{1}} = \PROFIT{\GROUP{\bm{{\hat{f}}},\EPS}{i}}$, for all $i \in [g_{\EPS}]$. \\ \\

\begin{lemma}\label{lem:profit_groups_perfect_predictions}
Fix any set $S = \{s_1,s_2,\ldots,s_k\} \subseteq (0,1]$, any $\EPS \in (0,1)$, and any instance, $(\SIGMA,\bm{{\hat{f}}})$, for $\DBC{S}{\mathcal{F}}$, with $\bm{{\hat{f}}} = \bm{{f}}$ and $\abs{\SIGMA} > m_{k,\EPS}^2 + m_{k,\EPS}$.
Then,
\begin{align*}
g_{\EPS} \cdot \PROFIT{\GROUP{\bm{{\hat{f}}},\EPS}{1}} \geqslant (1-\EPS) \cdot \OPT(\SIGMA).
\end{align*}
\end{lemma}
\begin{proof}
We show this by creating a solution, $N$, based on $\OPT[\SIGMA]$, such that
\begin{enumerate}[label = {(\roman*)}]
\item $\PROFIT{N} \geqslant \left(1-\frac{\EPS}{3}\right) \cdot \OPT(\SIGMA)$, and \label{item:N_compared_to_opt}
\item $g_{\EPS} \cdot \PROFIT{\GROUP{\bm{{\hat{f}}},\EPS}{1}} \geqslant \left(1-\frac{2\cdot\EPS}{3}\right) \cdot \PROFIT{N}$. \label{item:N_compared_to_groups}
\end{enumerate}
Since $\EPS \in (0,1)$, verifying~\ref{item:N_compared_to_opt} and~\ref{item:N_compared_to_groups}, implies that
\begin{align*}
g_{\EPS} \cdot \PROFIT{\GROUP{\bm{{\hat{f}}},\EPS}{1}} \geqslant \left(1-\frac{2\cdot\EPS}{3}\right) \cdot \left(1-\frac{\EPS}{3}\right) \cdot \OPT(\SIGMA) \geqslant (1-\EPS) \cdot \OPT(\SIGMA).
\end{align*}
\textbf{Construction of $\bm{{N}}$:}
Initially, let $N$ be a copy of $\OPT[\SIGMA]$.
Since $\OPT$ is non-wasteful, all bins in $\OPT[\SIGMA]$ are filled according to non-wasteful bin types.
For each non-wasteful bin type $t \in \mathcal{T}_S$, remove between $0$ and $g_{\EPS}$ bins of type $t$ from $N$, such that the number of bins of type $t$ is divisible by $g_{\EPS}$.
 
\textbf{Proof of~\ref{item:N_compared_to_opt}:}
Since $\OPT(\SIGMA) \geqslant \frac{n^\sigma}{\tau_S^m}$, Lemma~\ref{lem:num_completed_groups} implies that
\begin{align*}
\PROFIT{N} &\geqslant \OPT(\SIGMA) - (g_{\EPS}-1) \cdot \tau_S \geqslant \OPT(\SIGMA) - \frac{n^\sigma}{m_{\EPS}} \cdot \tau_S \\
&\geqslant \OPT(\SIGMA) - \OPT(\SIGMA) \cdot \frac{\tau_S \cdot \tau_S^m}{m_{\EPS}} \geqslant \left(1-\frac{\EPS}{3}\right) \cdot \OPT(\SIGMA).
\end{align*}
\textbf{Proof of~\ref{item:N_compared_to_groups}:}
Since the number of occurrences of each bin type in $N$ is divisible by $g_{\EPS}$, we may consider $N$ as $g_{\EPS}$ copies of a smaller covering $\overline{N}$.
Since we do not add any items when creating $N$, and thus $\overline{N}$, we get that $n_i^{\overline{N}} \leqslant \left\lfloor \frac{n_i^\SIGMA}{g_{\EPS}} \right\rfloor$, for all $i \in [k]$. 
Then, for all $i \in [k]$,
\begin{align*}
n_i^{\overline{N}} \leqslant \left\lfloor \frac{n_i^\SIGMA}{g_{\EPS}} \right\rfloor \leqslant \left\lfloor \frac{n_i^\SIGMA}{\left\lfloor \frac{n^\SIGMA}{m_{k,\EPS}}\right\rfloor} \right\rfloor \leqslant\left\lfloor \frac{n_i^\SIGMA}{\frac{n^\SIGMA}{m_{k,\EPS}} - 1} \right\rfloor = \left\lfloor \frac{n_i^\SIGMA }{\frac{n^\SIGMA - m_{k,\EPS}}{m_{k,\EPS}}} \right\rfloor = \left\lfloor n_i^\SIGMA \cdot \frac{m_{k,\EPS}}{n^\SIGMA - m_{k,\EPS}} \right\rfloor.
\end{align*}
Using that $\frac{m_{k,\EPS}}{n^\SIGMA + m_{k,\EPS}} = \frac{m_{k,\EPS}}{n^\SIGMA} + \frac{m_{k,\EPS}^2}{n^\SIGMA \cdot (n^\SIGMA - m_{k,\EPS})}$, and that $n^\SIGMA > m_{k,\EPS}^2 +m_{k,\EPS}$,
\begin{align*}
n_i^{\overline{N}} \leqslant \left\lfloor \frac{n_i^\SIGMA \cdot m_{k,\EPS}}{n^\SIGMA} + \frac{m_{k,\EPS}^2}{n^\SIGMA - m_{k,\EPS}} \right\rfloor \leqslant \left\lfloor \frac{n_i^\SIGMA \cdot m_{k,\EPS}}{n^\SIGMA} \right\rfloor + 1 = \lfloor f_i \cdot m_{k,\EPS} \rfloor + 1.
\end{align*}
Hence, $\overline{N}$ contains at most one more item of size $s_i$ than $\GROUP{\bm{{\hat{f}}},\EPS}{j}$, for all $i \in [k]$, and all $j \in [g_{\EPS}]$.
Then, for all $j \in [g_{\EPS}]$,
\begin{align}\label{eq:bound_profit_G}
\PROFIT{\GROUP{\bm{{\hat{f}}},\EPS}{j}} \geqslant \PROFIT{\overline{N}} - k.
\end{align}
Next, we lower bound $\PROFIT{\overline{N}}$.
Since $\OPT(\SIGMA) \geqslant \frac{n^{\sigma}}{\tau_S^m}$,
\begin{align*}
\PROFIT{\overline{N}} &= \frac{\PROFIT{N}}{g_{\EPS}} \geqslant \frac{\left(1-\frac{\EPS}{3}\right) \cdot \OPT(\SIGMA)}{g_{\EPS}} \geqslant \frac{\left(1-\frac{\EPS}{3}\right) \cdot n^\SIGMA}{\tau_S^m \cdot g_{\EPS}} \geqslant \frac{\left(1-\frac{\EPS}{3}\right) \cdot n^\SIGMA}{\tau_S^m \cdot \frac{n^\SIGMA}{m_{\EPS}}} \\
&= \frac{\left(1-\frac{\EPS}{3}\right) \cdot m_{\EPS}}{\tau_S^m} \geqslant \frac{\left(1-\frac{\EPS}{3}\right) \cdot \frac{3 \cdot \tau_S \cdot \tau_S^m}{\EPS}}{\tau_S^m} \geqslant \frac{\left(1-\frac{\EPS}{3}\right) \cdot 3 \cdot \tau_S}{\EPS} \geqslant \frac{\left(1-\frac{\EPS}{3}\right) \cdot k}{\frac{\EPS}{3}}.
\end{align*}
Hence, $k \leqslant \frac{\frac{\EPS}{3} \cdot \PROFIT{\overline{N}}}{1-\frac{\EPS}{3}}$, and so, by Equation~\eqref{eq:bound_profit_G},
\begin{align*}
\PROFIT{\GROUP{\bm{{\hat{f}}},\EPS}{j}}  \geqslant \PROFIT{\overline{N}} - \frac{\frac{\EPS}{3} \cdot \PROFIT{\overline{N}}}{1-\frac{\EPS}{3}} \geqslant \left(1-\frac{2\cdot\EPS}{3}\right) \cdot \PROFIT{\overline{N}}.
\end{align*}
Since $\PROFIT{N} = g_{\EPS} \cdot \PROFIT{\overline{N}}$ and $\PROFIT{\GROUP{\bm{{\hat{f}}},\EPS}{j}} = \PROFIT{\GROUP{\bm{{\hat{f}}},\EPS}{1}}$, for all $j \in [g_{\EPS}]$,
\begin{align*}
g_{\EPS} \cdot \PROFIT{\GROUP{\bm{{\hat{f}}},\EPS}{1}} \geqslant g_{\EPS} \cdot \left(1-\frac{2\cdot\EPS}{3}\right) \cdot \PROFIT{\overline{N}} = \left(1-\frac{2\cdot\EPS}{3}\right)  \cdot \PROFIT{N}.
\end{align*}
\end{proof}

We are now ready to state and prove the main result of this section.

\begin{theorem}\label{thm:consistency_of_profilecovering}
For any set $S = \{s_1,s_2,\ldots,s_k\} \subseteq (0,1]$, and any $\EPS \in (0,1)$, there exists a constant, $b$, such that for all instances $(\SIGMA,\bm{{\hat{f}}})$, with $\bm{{f}} = \bm{{\hat{f}}}$,
\begin{align*}
\PC{\EPS}(\SIGMA,\bm{{\hat{f}}}) \geqslant (1 - \EPS) \cdot \OPT(\SIGMA) - b.
\end{align*}
That is, $\PC{\EPS}$ is a $(1-\EPS)$-consistent algorithm for $\DBC{S}{\mathcal{F}}$.
\end{theorem}
\begin{proof}
We set $b = m_{k,\EPS}^2 + m_{k,\EPS}$.
If $\abs{\SIGMA} \leqslant m_{k,\EPS}^2 + m_{k,\EPS}$, the right-hand size is non-positive, and the left-hand side is non-negative, and so we are done.
On the other hand, if $\abs{\SIGMA} > m_{k,\EPS}^2 + m_{k,\EPS}$, then, by Lemma~\ref{lem:profit_groups_perfect_predictions}, \\
\begin{minipage}{\textwidth}    
\begin{align*}
\PC{\EPS}(\SIGMA,\bm{{\hat{f}}}) \geqslant g_{\EPS} \cdot \PROFIT{\GROUP{\bm{{\hat{f}}},\EPS}{1}} \geqslant (1-\EPS) \cdot \OPT(\SIGMA). 
\end{align*}
\end{minipage} \vspace*{-.2cm}
\end{proof}

Consider the instance $(\SIGMA^n,\bm{{\hat{f}}})$ where $\SIGMA^n = \left\langle \frac{1}{k} \right\rangle^n$ and $\bm{{\hat{f}}}$ predicts that half of the items are of size $\frac{1}{k}$, and half of the items are of size $\frac{k-1}{k}$, a wrong prediction for $\SIGMA^n$.
Based on the predictions $\bm{{\hat{f}}}$, $\PC{\EPS}$ creates $\left\lfloor \frac{m_{k,\EPS}}{2} \right\rfloor$ bins that contain placeholders for one item of size $\frac{1}{k}$, and one item of size $\frac{k-1}{k}$.
Since no item of size $\frac{k-1}{k}$ arrives, $\PC{\EPS}$ never covers a bin, and since $\OPT(\SIGMA^n) = \left\lfloor \frac{n}{k} \right\rfloor$, $\PC{\EPS}$ is not robust.
Hence, $\PC{\EPS}$ is not a robust algorithm for $\DBC{k}{\mathcal{F}}$.
In the next section, we create a strategy for improving the robustness of $\PC{\EPS}$.

\subsection{Robustifying $\PC{\EPS}$}
For each purely online algorithm, $\ALG$ ($\ALG$ could be $\DNF$), we create a family of \emph{hybrid algorithms} that merges $\PC{\EPS}$ with $\ALG$ to improve on the robustness of $\PC{\EPS}$.
Formally, for a fixed algorithm, $\ALG$, we create the family $\{\HY{\ALG}{\TL}{\EPS}\}_{\TL,\EPS}$, of hybrid algorithms, parametrized by $\EPS \in (0,1)$ and a \emph{trust level}, $\TL = \frac{\KAPPA}{\LAMBDA} \in \QQ^+$.
For our applications, we assume that $\TL$ is given as a fraction, $\TL = \frac{\KAPPA}{\LAMBDA}$.
For a fixed purely online algorithm, $\ALG$, $\TL \in \QQ^+$, and $\EPS \in (0,1)$, we give the pseudo-code for $\HY{\ALG}{\TL}{\EPS}$ in Algorithm~\ref{alg:hybrid}.

In words, upon receiving an item $a$, $\HY{\ALG}{\TL}{\EPS}$ counts the number of occurences of $a$, denoted $c_a$, and if $c_a\ (\on{mod}\, \LAMBDA) \leqslant \LAMBDA - \KAPPA - 1$, we ask $\ALG$ to place $a$ in a bin, that only $\ALG$ places items into, and otherwise, we ask $\PC{\EPS}$ to place $a$ in a bin that only $\PC{\EPS}$ places items into.

\begin{algorithm}
\caption{$\HY{\ALG}{\TL}{\EPS}$}
\begin{algorithmic}[1]
\STATE \textbf{Input:} An instance for $\DBC{S}{\mathcal{F}}$, $(\SIGMA,\bm{{\hat{f}}})$
\STATE Determine $\KAPPA,\LAMBDA \in \ZZ^+$ such that $\TL = \frac{\KAPPA}{\LAMBDA}$
\STATE Run Lines~\ref{line:pc_initialization_start}-\ref{line:pc_initialization_end} of $\PC{\EPS}$ (see Algorithm~\ref{alg:profile_covering}), given the prediction $\bm{{\hat{f}}}$ 
\STATE Run initialization part of $\ALG$, if such exists
\FORALL {$i \in [k]$}
	\STATE $c_{s_i} \leftarrow 0$
\ENDFOR
\WHILE {receiving items, $a$,}
	\STATE $j \leftarrow c_{a}\ (\on{mod}\, \LAMBDA)$ \COMMENT{$a \in \SIGMA_{j+1}$}
	\IF {$j \leqslant \LAMBDA - \KAPPA - 1$}
		\STATE Ask $\ALG$ to distribute $a$
	\ELSE[$\LAMBDA - \KAPPA \leqslant j \leqslant \LAMBDA - 1$]
		\STATE Ask $\PC{\EPS}$ to distribute $a$ \COMMENT{See Lines~\ref{line:pc_distribution_start}-\ref{line:pc_distribution_end} in Algorithm~\ref{alg:profile_covering}}
	\ENDIF
	\STATE $c_{a} \leftarrow c_{a} + 1$
\ENDWHILE
\end{algorithmic}
\label{alg:hybrid}
\end{algorithm}

Towards the analysis of $\HY{\ALG}{\TL}{\EPS}$, we associate, to any instance, $\sigma$, of $\DBC{S}{}$, a $(\LAMBDA+1)$-tuple, $(\SIGMA_1,\SIGMA_2,\ldots,\SIGMA_\LAMBDA,\SIGMA_e) = \LS(\SIGMA)$ (see Algorithm~\ref{alg:lambda_splitting}) called the \emph{$\LAMBDA$-splitting of $\SIGMA$}.

In words, the $\LAMBDA$-splitting of $\SIGMA$ is created by, for each $s_i \in S$, counting the number of items in $\SIGMA$ of size $s_i$. 
Processing the items one-by-one, in the order they appear in $\SIGMA$, then, when processing an item, $a$, we place $a$ in $\SIGMA_{i+1}$ if $c_{a}\ (\on{mod}\, \LAMBDA) \equiv i$.
After processing all items in $\SIGMA$, then, for all $s_i \in S$, we compute the number of items of size $s_i$ in each $\SIGMA_j$, for $i \in [k]$ and $j \in [\LAMBDA]$.
If there are equally many items of size $s_i$ in all $\SIGMA_j$, we are done. 
If, on the other hand, there exists some $i \in [k]$ and some $j \in [\LAMBDA]$ such that $\SIGMA_1,\SIGMA_2,\ldots,\SIGMA_j$ contains one more item of size $s_i$ than $\SIGMA_{j+1},\SIGMA_{j+2},\ldots,\SIGMA_{\LAMBDA}$, then we remove one item of size $s_i$ from all of $\SIGMA_1,\SIGMA_2,\ldots,\SIGMA_j$, and place it in $\SIGMA_e$ instead.

By construction, the $\LAMBDA$-splitting of $\SIGMA$ decomposes $\SIGMA$ into $\LAMBDA$ smaller instances, $\SIGMA_i$ for $i \in [\LAMBDA]$, that all contain the same multiset of items, but possibly in different orders, and an \emph{excess} instance $\SIGMA_e$, which contain the remaining items from $\SIGMA$.
By construction, $\abs{\SIGMA_e} \leqslant (\LAMBDA - 1) \cdot k$.

\begin{algorithm}
\caption{$\LS$}
\begin{algorithmic}[1]
\STATE \textbf{Input:} an instance for $\DBC{S}{}$, $\SIGMA$
\STATE $\SIGMA_e \leftarrow \emptyset$, and $\SIGMA_i \leftarrow \emptyset$, for all $i \in [\LAMBDA]$
\STATE $c_{s_i} \leftarrow 0$, for all $s_i \in S$.
\FOR {each item, $a \in \SIGMA$,}
	\STATE $i \leftarrow c_{a}\, (\on{mod}\, \LAMBDA)$
	\STATE $\SIGMA_{i+1} \leftarrow \SIGMA_{i+1} \cup \{a\}$
	\STATE $c_{a} \leftarrow c_{a} + 1$
\ENDFOR
\FOR {each $s_i \in S$} \label{line:start_check_lambda_splitting}
	\STATE $j \leftarrow c_{s_i}\, (\on{mod}\, \LAMBDA)$
	\IF {$j \neq 0$}
		\FOR[$\SIGMA_1,\SIGMA_2,\ldots,\SIGMA_j$ contain too many $s_i$'s] {$l=1,2,\ldots,j$}
			\STATE Remove the last item of size $s_i$ from $\SIGMA_l$ an add it to $\SIGMA_e$ \label{line:end_check_lambda_splitting}
		\ENDFOR
	\ENDIF
\ENDFOR
\STATE \textbf{Return:} $(\SIGMA_1,\SIGMA_2,\ldots,\SIGMA_\LAMBDA,\SIGMA_e)$.
\end{algorithmic}
\label{alg:lambda_splitting}
\end{algorithm}

\paragraph{Bounding the Performance of $\OPT$}

Throughout, given $\LAMBDA$ instances, $\SIGMA_1,\SIGMA_2,\ldots,\SIGMA_{\LAMBDA}$, we set $\bigcup_{i=1}^{\LAMBDA} \SIGMA_i = \langle \SIGMA_1, \SIGMA_2, \ldots, \SIGMA_{\LAMBDA} \rangle$.

\begin{observation}\label{lem:opt_on_split_instances_1}
Let $\SIGMA_1,\SIGMA_2,\ldots,\SIGMA_\LAMBDA$ be any instances for $\DBC{S}{}$, then 
\begin{align*}
\sum_{i=1}^\LAMBDA \OPT(\SIGMA_i) \leqslant \OPT\left( \bigcup_{i=1}^\LAMBDA \SIGMA_i \right).
\end{align*}
\end{observation}

\begin{lemma}\label{lem:bound_for_opt_on_split_instances}
Let $S = \{s_1,s_2,\ldots,s_k\} \subseteq (0,1]$ be any finite set, let $\SIGMA$ by any instance of $\DBC{S}{}$, and let $(\SIGMA_1,\SIGMA_2,\ldots,\SIGMA_\LAMBDA,\SIGMA_e)$ be the $\LAMBDA$-splitting of $\SIGMA$.
Then,
\begin{align*}
\OPT(\SIGMA) \leqslant \sum_{i=1}^\LAMBDA \OPT(\SIGMA_i)  + (\LAMBDA - 1) \cdot (k + \tau_S).
\end{align*}
\end{lemma}

\begin{proof}
We split this proof into two parts, by showing that
\begin{enumerate}[label = {(\roman*)}]
\item $\OPT\left( \bigcup_{i=1}^\LAMBDA \SIGMA_i \right) \leqslant \sum_{i=1}^{\LAMBDA} \OPT(\SIGMA_i) + (\LAMBDA - 1) \cdot \tau_S$, and \label{item:opt_on_split_instances_part_1}
\item $\OPT(\SIGMA) \leqslant \OPT\left( \bigcup_{i=1}^\LAMBDA \SIGMA_i \right) + (\LAMBDA -1) \cdot k$. \label{item:opt_on_split_instances_part_2}
\end{enumerate}
\textbf{Proof of~\ref{item:opt_on_split_instances_part_1}:}
We use a similar strategy as in the proof of Theorem~\ref{thm:consistency_of_profilecovering}.
To this end, let $N$ be the solution obtained by removing at most $\LAMBDA-1$ bins of each non-wasteful bin type from a copy of $\OPT\left[ \bigcup_{i=1}^\LAMBDA \SIGMA_i \right]$ (recall that $\OPT$ is non-wasteful) such that the number of each bin type in $N$ is divisible by $\LAMBDA$. 
Then,
\begin{align*}
\PROFIT{N} \geqslant \OPT\left( \bigcup_{i=1}^\LAMBDA \SIGMA_i \right) - (\LAMBDA-1) \cdot \tau_S.
\end{align*}
Hence, it remains to compare the profit of $\bigcup_{i=1}^\LAMBDA \OPT[\SIGMA_i]$ to $\PROFIT{N}$.
Since $\SIGMA_1,\SIGMA_2,\ldots,\SIGMA_\LAMBDA$ all contain the same multiset of items (but possibly in a different order), we get that $\OPT(\SIGMA_i) = \OPT(\SIGMA_j)$, for all $i,j \in [\LAMBDA]$.
Further, by construction, $N$ is the union of $\LAMBDA$ identical smaller coverings, $\overline{N}$, for which $n_i^{\overline{N}} \leqslant n_i^{\SIGMA_i}$, for all $i \in [k]$.
Therefore, $\OPT(\SIGMA_i) \geqslant \PROFIT{\overline{N}}$, for all $i \in [k]$, and we can write
\begin{align*}
\sum_{i=1}^\LAMBDA \OPT(\SIGMA_i) = \LAMBDA \cdot \OPT(\SIGMA_1) \geqslant \LAMBDA \cdot \PROFIT{\overline{N}} = \PROFIT{N},
\end{align*}
which completes the proof of~\ref{item:opt_on_split_instances_part_1}.

\textbf{Proof of~\ref{item:opt_on_split_instances_part_2}:}
Since $\abs{\SIGMA_e} \leqslant (\LAMBDA -1)\cdot k$, we can write
\begin{align*}
\OPT\left( \bigcup_{i=1}^\LAMBDA \SIGMA_i \right) \geqslant \OPT(\SIGMA) - (\LAMBDA-1)\cdot k.
\end{align*}
Adding $(\LAMBDA-1)\cdot k$ to both sides verifies~\ref{item:opt_on_split_instances_part_2}, and thus completes the proof.
\end{proof}

\paragraph{Bound on the Performance of $\PC{\EPS}$}
We compare the number of bins covered by $\PC{\EPS}$ on a subset of the instances in the $\LAMBDA$-splitting of an instance, $\SIGMA$, to that of $\OPT$ on $\SIGMA$.
To this end, observe that if $\SIGMA$ is a $\DBC{S}{}$-instance, where $S = \{s_1,s_2,\ldots,s_k\} \subseteq (0,1]$, and $(\SIGMA_1,\SIGMA_2,\ldots,\SIGMA_\LAMBDA,\SIGMA_e)$ is the $\LAMBDA$-splitting of $\SIGMA$, then $n_j^{\SIGMA_i} = \left\lfloor \frac{n_j^\SIGMA}{\LAMBDA}\right\rfloor$, for all $j \in [k]$ and all $i \in [\LAMBDA]$ (see Lines~\ref{line:start_check_lambda_splitting}--\ref{line:end_check_lambda_splitting} in Algorithm~\ref{alg:lambda_splitting}).

\begin{lemma}\label{lem:PC_on_split_instances}
Fix any set $S = \{s_1,s_2,\ldots,s_k\} \subseteq (0,1]$, any $\EPS \in (0,1)$, and any instance $(\SIGMA,\bm{{\hat{f}}})$ of $\DBC{S}{}$, for which $\bm{{f}} = \bm{{\hat{f}}}$, and let $(\SIGMA_1,\SIGMA_2,\ldots,\SIGMA_\LAMBDA,\SIGMA_e)$ be the $\LAMBDA$-splitting of $\SIGMA$, for some $\LAMBDA \in \ZZ^+$.
Then, for any $j \in \ZZ^+$, with $j \leqslant \LAMBDA$, there exists a constant $b$ such that
\begin{align*}
\PC{\EPS}\left( \left(\bigcup_{i = \LAMBDA - j + 1}^\LAMBDA \SIGMA_i \right) , \bm{{\hat{f}}} \right) \geqslant \frac{j \cdot (1-\EPS) \cdot \OPT(\SIGMA)}{\LAMBDA} - b.
\end{align*}
\end{lemma}
\begin{proof}
Let $\tilde{\SIGMA_j} = \bigcup_{i=\LAMBDA - j + 1}^\LAMBDA \SIGMA_i$, and set $b = m_{k,\EPS}^2 + m_{k,\EPS} + k \cdot j$.
If $\abs{\SIGMA} \leqslant b$, the right-hand side is non-positive, and the left-hand side is non-negative, and the lemma's statement trivially follows. 

Hence, assume that $\abs{\SIGMA} > b$.
Let $C = \PC{\EPS}[\SIGMA,\bm{{\hat{f}}}]$, and let $g_{\EPS}$ be the number of groups, $\GROUP{\bm{{\hat{f}}},\EPS}{i}$, that $\PC{\EPS}$ completes on instance $(\SIGMA,\bm{{\hat{f}}})$.
By Lemma~\ref{lem:profit_groups_perfect_predictions}, $g_{\EPS} \cdot \PROFIT{\GROUP{\bm{{\hat{f}}},\EPS}{1}} \geqslant (1-\EPS) \cdot \OPT(\SIGMA)$.
Since $\GROUP{\bm{{\hat{f}}},\EPS}{i}$ is only dependent on $\EPS$, $S$, and $\bm{{\hat{f}}}$, $\PC{\EPS}$ creates the same groups, $\GROUP{\bm{{\hat{f}}},\EPS}{i}$, on instance $(\SIGMA,\bm{{\hat{f}}})$ as on instance $(\tilde{\SIGMA_j} ,\bm{{\hat{f}}})$. 
In the following, we prove a lower bound for the number of groups that $\PC{\EPS}$ completes on instance $(\tilde{\SIGMA_j} ,\bm{{\hat{f}}})$, as a function of $g_{\EPS}$.

Since $C$ completely covers $g_{\EPS}$ copies of $\GROUP{\bm{{\hat{f}}},\EPS}{i}$, then $n_i^\SIGMA \geqslant g_{\EPS} \cdot \lfloor f_i^\SIGMA \cdot m_{k,\EPS} \rfloor$ for all $i \in [k]$.
Since each $\SIGMA_i$ contains exactly $\left\lfloor \frac{n_i^\SIGMA}{\LAMBDA} \right\rfloor$ items of size $s_i$, then
\begin{align*}
n_i^{\tilde{\SIGMA}_j} \geqslant \frac{j \cdot n_i^\SIGMA}{\LAMBDA} - j \geqslant \frac{j \cdot g_{\EPS}}{\LAMBDA} \cdot \lfloor f_i^\SIGMA \cdot m_{k,\EPS} \rfloor - j \geqslant \left\lfloor \frac{j \cdot g_{\EPS}}{\LAMBDA} \right\rfloor \cdot \lfloor f_i^\SIGMA \cdot m_{k,\EPS} \rfloor - j.
\end{align*}
This implies that, $\PC{\EPS}$ fills in all placeholders for items of size $s_i$ in $\left\lfloor \frac{j \cdot g_{\EPS}}{\LAMBDA} \right\rfloor$ groups, except at most $j$, on instance $(\tilde{\SIGMA}_j,\bm{{\hat{f}}})$, for all $i \in [k]$.
Hence,
\begin{align*}
\PC{\EPS}(\tilde{\SIGMA}_j,\bm{{\hat{f}}}) \geqslant \left\lfloor \frac{j \cdot g_{\EPS}}{\LAMBDA} \right\rfloor \cdot \PROFIT{\GROUP{\bm{{\hat{f}}},\EPS}{i}} - k \cdot j \geqslant \left( \frac{j \cdot g_{\EPS}}{\LAMBDA} - 1 \right) \cdot \PROFIT{\GROUP{\bm{{\hat{f}}},\EPS}{i}} - k \cdot j.
\end{align*}
Since $\PROFIT{\GROUP{\bm{{\hat{f}}},\EPS}{i}} \leqslant m_{k,\EPS}$, we get that
\begin{align*}
\PC{\EPS}(\tilde{\SIGMA}_j,\bm{{\hat{f}}}) \geqslant \frac{j \cdot g_{\EPS}}{\LAMBDA} \cdot \PROFIT{\GROUP{\bm{{\hat{f}}},\EPS}{i}} - k \cdot j - m_{k,\EPS} \geqslant \frac{j \cdot (1- \EPS) \cdot \OPT(\SIGMA)}{\LAMBDA} - b,
\end{align*}
which completes the proof.
\end{proof}

\subsection{A Trust-Parametrized Family of Hybrid Algorithms}

In what follows, we wrap up the analysis of $\HY{\ALG}{\TL}{\EPS}$ by stating and proving the main results of this section.
By construction, $\HY{\ALG}{\TL}{\EPS}$ (see Algorithm~\ref{alg:hybrid}) distributes the items that arrive between $\PC{\EPS}$ and $\ALG$ in a way determined by $\TL$.
Whenever $\TL$ becomes close to $1$, $\HY{\ALG}{\TL}{\EPS}$ gives a larger fraction of items to $\PC{\EPS}$, and when $\TL$ gets close to $0$, $\HY{\ALG}{\TL}{\EPS}$ gives more items to $\ALG$.
In particular, $\HY{\ALG}{1}{\EPS} = \PC{\EPS}$, and $\HY{\ALG}{0}{\EPS} = \ALG$.
Clearly, $\HY{\ALG}{\TL}{\EPS}$ cannot create a perfect $\LAMBDA$-splitting online, since it cannot correctly identify the items that are placed in $\SIGMA_e$.
It can, however, get sufficiently close.

\begin{theorem}\label{thm:consistency_hybrid}
For any finite set $S = \{s_1,s_2,\ldots,s_k\} \subseteq (0,1]$, any purely online $\DBC{S}{}$-algorithm, $\ALG$, any $c \leqslant \CR{\ALG}$, any $\EPS \in (0,1)$, and any $\TL \in \QQ^+$, there exists a constant $b \in \ZZ^+$, such that for all instances $(\SIGMA,\bm{{\hat{f}}})$ of $\DBC{S}{}$, with $\bm{{f}} = \bm{{\hat{f}}}$,
\begin{align*}
\HY{\ALG}{\TL}{\EPS}(\SIGMA,\bm{{\hat{f}}}) \geqslant (\TL \cdot (1-\EPS) + (1-\TL) \cdot c) \cdot \OPT(\SIGMA) - b.
\end{align*}
\end{theorem}
\begin{proof}
Let $b_{\ALG}$ be the additive constant of $\ALG$, $b_{\PC{\EPS}} = m_{k,\EPS}^2 + m_{k,\EPS} + k \cdot j$.
Then, we set $b = b_{\ALG} + b_{\PC{\EPS}} + (\LAMBDA - 1) \cdot (k + \tau_S)$.
If $\abs{\SIGMA} \leqslant b$, the result follows trivially.
Hence, assume that $\abs{\SIGMA} > b$.

Let $(\SIGMA_1,\SIGMA_2,\ldots,\SIGMA_\LAMBDA,\SIGMA_e)$ be the $\LAMBDA$-splitting of $\SIGMA$.
Throughout this proof, let $\SIGMA_e^{\ALG}$ and $\SIGMA_e^{\PC{\EPS}}$ be the collection of instances from $\SIGMA_e$ that $\ALG$ and $\PC{\EPS}$ receive, respectively.
Then, by definition of $\HY{\ALG}{\TL}{\EPS}$,
\begin{align*}
\HY{\ALG}{\TL}{\EPS}[\SIGMA,\bm{{\hat{f}}}] = \ALG\left[ \left(\bigcup_{i=1}^{\LAMBDA - \KAPPA} \SIGMA_i \right) \cup \SIGMA_e^{\ALG} \right] \cup \PC{\EPS}\left[ \left(\bigcup_{i=\LAMBDA-\KAPPA + 1}^\LAMBDA \SIGMA_i \right) \cup \SIGMA_e^{\PC{\EPS}} , \bm{{\hat{f}}} \right]. 
\end{align*}
Hence,
\begin{align*}
\HY{\ALG}{\TL}{\EPS}(\SIGMA,\bm{{\hat{f}}}) \geqslant \ALG\left( \bigcup_{i=1}^{\LAMBDA-\KAPPA} \SIGMA_i \right) + \PC{\EPS}\left( \left (\bigcup_{i=\LAMBDA-\KAPPA+1}^\LAMBDA \SIGMA_i \right) , \bm{{\hat{f}}} \right).
\end{align*}
Set $b' = b_{\ALG} + b_{\PC{\EPS}}$. 
Then, by $c$-competitiveness of $\ALG$ and Lemma~\ref{lem:PC_on_split_instances}, 
\begin{align*}
\HY{\ALG}{\TL}{\EPS}(\SIGMA,\bm{{\hat{f}}}) &\geqslant c \cdot \OPT\left( \bigcup_{i=1}^{\LAMBDA-\KAPPA} \SIGMA_i \right) + \TL \cdot (1-\EPS) \cdot \OPT(\SIGMA) - b'.
\end{align*}
Since $\sum_{i=1}^\LAMBDA \OPT(\SIGMA_i) \leqslant \OPT\left( \bigcup_{i=1}^\LAMBDA \SIGMA_i \right)$, by Observation~\ref{lem:opt_on_split_instances_1}, and $\OPT(\SIGMA_i) = \OPT(\SIGMA_j)$, for all $i,j \in [\LAMBDA]$,
\begin{align*}
\HY{\ALG}{\TL}{\EPS}(\SIGMA,\bm{{\hat{f}}}) &\geqslant c \cdot \left( \sum_{i=1}^{\LAMBDA - \KAPPA} \OPT(\SIGMA_i) \right) + \TL \cdot (1-\EPS) \cdot \OPT(\SIGMA) - b' \\
&= (1-\TL) \cdot c \cdot \left( \sum_{i=1}^{\LAMBDA} \OPT(\SIGMA_i) \right)  + \TL \cdot (1-\EPS) \cdot \OPT(\SIGMA) - b'.
\end{align*}
Then, by Lemma~\ref{lem:bound_for_opt_on_split_instances},

\scalebox{.9}{
\begin{minipage}{1.1\textwidth}    
\begin{align*}
\HY{\ALG}{\TL}{\EPS}(\SIGMA,\bm{{\hat{f}}}) &\geqslant (1-\TL) \cdot c \cdot (\OPT(\SIGMA) - (\LAMBDA - 1) \cdot (k + \tau_S))  + \TL \cdot (1-\EPS) \cdot \OPT(\SIGMA) - b' \\
&\geqslant ((1-\TL) \cdot c + \TL \cdot (1-\EPS)) \cdot \OPT(\SIGMA) - b.
\end{align*}
\end{minipage}
}

\end{proof}

The above shows an explicit formula for the consistency of $\HY{\ALG}{\TL}{\EPS}$ as a function of the trust-level, $\TL$, $\EPS \in (0,1)$, and the performance guarantee of $\ALG$.
A very similar proof shows a guarantee on the robustness of $\HY{\ALG}{\TL}{\EPS}$.

\begin{theorem}\label{thm:robustness_hybrid}
For any finite set $S = \{s_1,s_2,\ldots,s_k\} \subseteq (0,1]$, any purely online algorithm, $\ALG$, for $\DBC{S}{}$, any $c \leqslant \CR{\ALG}$, and any $\EPS$, there exists a constant $b \in \ZZ^+$, such that for all instances $(\SIGMA,\bm{{\hat{f}}})$,
\begin{align*}
\HY{\ALG}{\TL}{\EPS}(\SIGMA,\bm{{\hat{f}}}) \geqslant (1-\TL) \cdot c \cdot \OPT(\SIGMA) - b.
\end{align*}
\end{theorem}


\section{Stochastic Setting}

In this section, we consider a setting for $\DBC{S}{}$ where item sizes are generated independently at random from an unknown distribution. This setting has already been studied for the more restricted $\DBC{k}{}$ problem, where Csirik, Johnson and Kenyon used variants of the Bin Packing algorithm \textquotedblleft Sum-of-Squares\textquotedblright, first introduced  in~\cite{CJKSW99,CJKOSW06}, to develop algorithms for the problem.
Rather than designing algorithms that perform well in the worst case, they aimed to design algorithms that perform well in the average case.
Specifically, they develop an algorithm, called $SS^\ast$, with $\EAR{SS^\ast}(D) = 1$ (see Equation~\eqref{eq:asymptotic_expected_ratio} for the definition of \EAR{SS^\ast}), for all discrete distributions $D$ of $F_k$, with rational probabilities.

In this section, we use a PAC-learning bound for learning frequencies in discrete distributions to derive a family of algorithms called \emph{purely online profile covering} ($\{\POPC{\EPS}{\DELTA}\}_{\EPS,\DELTA}$). These algorithms are parametrized by two real numbers $\EPS,\DELTA \in (0,1)$, satisfying that, for \emph{all} finite sets $S = \{s_1,s_2,\ldots,s_k\} \subseteq (0,1]$, there exists a constant $b \in \RR^+$, such that for all discrete distributions $D = \{p_1,p_2,\ldots,p_k\}$ of $S$, allowing irrational probabilities,
\begin{align}\label{eq:expected_ratio_POPC}
P\left( \POPC{\EPS}{\DELTA}(\SIGMA_n(D)) \geqslant (1-\EPS) \cdot \OPT(\SIGMA_n(D)) - b \right) \geqslant 1-\DELTA,
\end{align}
Observe that this guarantee is true, even for adversarial $S$ and $D$.
Clearly, Equation~\eqref{eq:expected_ratio_POPC} directly implies that
\begin{align}\label{eq:asymptotic_expected_ratio_POPC}
P(\EAR{\POPC{\EPS}{\DELTA}}(D) \geqslant 1- \EPS ) \geqslant 1- \DELTA.
\end{align}
The guarantee from Equation~\eqref{eq:expected_ratio_POPC} is, however, stronger than the from Equation~\eqref{eq:asymptotic_expected_ratio_POPC}, in that the additive term in Equation~\eqref{eq:expected_ratio_POPC} is constant, whereas the additive term for $\POPC{\EPS}{\DELTA}$ in Equation~\eqref{eq:asymptotic_expected_ratio_POPC} may be a function of $n$. As pointed out in~\cite{BF20}, having only constant loss before giving a multiplicative performance guarantee is a desirable property.

We formalize the strategy of $\POPC{\EPS}{\DELTA}$ in Algorithm~\ref{alg:popc}.
In words; the algorithm works by defining a ``sample size", $\THRESHOLDALG$, as a function of $k$, $\EPS$ and $\DELTA$. Intuitively, observing $\THRESHOLDALG$ items from the input prefix is sufficient to make predictions about the frequency of items with respect to $D$ that are $\EPS$-accurate with confidence $1-\DELTA$.
We prove this in Section~\ref{sec:pac}.
While learning $D$, $\POPC{\EPS}{\DELTA}$ places the first $\THRESHOLDALG$ items using the (purely online) $\DNF$ strategy while observing the item frequencies. 
After placing the first $\THRESHOLDALG$, the algorithm uses the observed frequencies to make an estimate - prediction - about the item frequencies and applies $\PC{\frac{\EPS}{2}}$ to place the remaining items.

\begin{algorithm}
\caption{$\POPC{\EPS}{\DELTA}$}
\begin{algorithmic}[1]
\STATE \textbf{Input:} A $\DBC{S}{}$-instance, $\SIGMA$
\STATE $ss \leftarrow 0$ \COMMENT{Sample size}
\STATE Compute $\tau_S$, $\tau_S^m$, and $k = \abs{S}$
\STATE $m_{\frac{\EPS}{2}} \leftarrow \lceil 6 \cdot \tau_S \cdot \tau_S^m \cdot \EPS^{-1} \rceil$
\STATE $m_{k,\frac{\EPS}{2}} \leftarrow m_{\frac{\EPS}{2}} + k$
\STATE $\THRESHOLDALG \leftarrow \on{max}\left\{16 \cdot k \cdot (m_{k,\frac{\EPS}{2}} + 1)^2, 32 \cdot (m_{k,\frac{\EPS}{2}} + 1)^2 \cdot \ln \left( \frac{2}{1 - \sqrt{1 - \DELTA}} \right) \right\}$
\FORALL {$i \in [k]$}
	\STATE $c_{s_i} \leftarrow 0$ \COMMENT{Number of items of size $s_i$}
\ENDFOR
\WHILE {receiving items, $a$, \textbf{and} $ss < \THRESHOLDALG$}
	\STATE $c_a \leftarrow c_a + 1$
	\STATE Place $a$ in a $\DNF$-marked bin using $\DNF$
	\STATE $ss \leftarrow ss + 1$
\ENDWHILE
\FOR {$i =1,2,\ldots,k$} \label{line:create_preds_start}
	\STATE $\hat{f}_i^\THRESHOLDALG = \frac{c_{s_i}}{\THRESHOLDALG}$ 
\ENDFOR
\STATE $\bm{{\hat{f}^\THRESHOLDALG}} = \left( \hat{f}_1^\THRESHOLDALG , \hat{f}_2^\THRESHOLDALG,\ldots,\hat{f}_k^\THRESHOLDALG \right)$ \label{line:create_preds_end}
\STATE Run Lines~\ref{line:pc_initialization_start}-\ref{line:pc_initialization_end} of $\PC{\frac{\EPS}{2}}$ (see Algorithm~\ref{alg:profile_covering}), given the prediction $\bm{{\hat{f}^\THRESHOLDALG}}$ 
\WHILE {receiving items, $a$,}
	\STATE Place $a$ using $\PC{\frac{\EPS}{2}}$ \COMMENT{See Lines~\ref{line:pc_distribution_start}-\ref{line:pc_distribution_end} in Algorithm~\ref{alg:profile_covering}}
\ENDWHILE
\end{algorithmic}
\label{alg:popc}
\end{algorithm}

Before formalizing and proving the claim from Equation~\eqref{eq:expected_ratio_POPC}, we review a PAC-learning bound for learning frequencies in discrete distributions~\cite{C20}.

\subsection{A PAC-Learning Bound for Learning Frequencies}\label{sec:pac}

It is well-known~\cite{C20}, that probabilities in discrete distributions are PAC-learnable~\cite{SB14}.

\begin{proposition}\label{prop:pac}
For any finite set $S = \{s_1,s_2,\ldots,s_k\} \subseteq (0,1]$, there exists an algorithm, $\PACALG$, and a map $\THRESHOLD{\PACALG}\colon \RR^+ \times (0,1) \rightarrow \ZZ^+$, such that for any $\GAMMA \in \RR^+$, any $\DELTA \in (0,1)$, any (unknown) discrete distribution $D = \{p_1,p_2,\ldots,p_k\}$ of $S$, and any $n \geqslant \THRESHOLD{\PACALG}(\GAMMA,\DELTA)$, letting $\{X_i\}_{i=1}^n$ be a sequence of independent identically distributed random variables, with $X_i \sim D$,
\begin{align*}
P\left( L^1(\PACALG(X_1,X_2,\ldots,X_n) , D) \leqslant \GAMMA \right) \geqslant 1 - \DELTA,
\end{align*}
where $L^1$ is the usual $L^1$-distance.
For learning frequencies in discrete distributions, $\PACALG$ is the algorithm which outputs the predicted distribution:
\begin{align*}
\PACALG(X_1,X_2,\ldots,X_n) = \left\{ \hat{p}_i\ \bigg|\ i \in [k] \text{ and } \hat{p}_i = \frac{1}{n} \cdot \sum_{j=1}^n \mathds{1}_{\{s_i\}}(X_j) \right\},
\end{align*}
and, for any $\GAMMA \in \RR^+$ and $\DELTA \in (0,1)$, the map $\THRESHOLD{\PACALG}$ is given by
\begin{align*}
\THRESHOLD{\PACALG}(\GAMMA,\DELTA) = \on{max}\left\{ \frac{4 \cdot k}{\GAMMA^2} , \frac{8}{\GAMMA^2} \cdot \ln\left( \frac{2}{\DELTA} \right) \right\}.
\end{align*}
\end{proposition}

\subsection{Analysis of $\POPC{\EPS}{\DELTA}$}

We formalize and prove the claim from Equation~\eqref{eq:expected_ratio_POPC}:

\begin{theorem}\label{thm:expected_ratio_popc}
For all finite sets $S = \{s_1,s_2,\ldots,s_k\} \subset (0,1]$, and all $\EPS, \DELTA \in (0,1)$, there exists a constant $b \in \ZZ^+$, such that for all discrete distributions $D = \{p_1,p_2,\ldots,p_k\}$ of $S$, and all $n \in \ZZ^+$,
\begin{align*}
P\left( \POPC{\EPS}{\DELTA}(\SIGMA_n(D)) \geqslant (1-\EPS) \cdot \OPT(\SIGMA_n(D)) - b \right) \geqslant 1-\DELTA,
\end{align*}
where $\SIGMA_n(D) = \langle X_1,X_2, \ldots,X_n\rangle$, and $\{X_i\}_{i=1}^n$ is a sequence of independent identically distributed random variables with $X_i \sim D$, for all $i \in [n]$.
\end{theorem}
\begin{proof}
Set $\THRESHOLDALG = \on{max}\left\{16 \cdot k \cdot (m_{k,\frac{\EPS}{2}} + 1)^2, 32 \cdot (m_{k,\frac{\EPS}{2}} + 1)^2 \cdot \ln \left( \frac{2}{1 - \sqrt{1 - \DELTA}} \right) \right\}$, and $b = \on{max}\{2 \cdot \THRESHOLDALG,m_{k,\frac{\EPS}{2}}^2 + m_{k,\frac{\EPS}{2}} + \THRESHOLDALG\}$, and observe that $b$ is independent of the input length $n$.
By similar arguments as above, we assume that $n \geqslant b$.

For ease of notation, we set $\tilde{\EPS} = \frac{\EPS}{2}$.

Throughout this proof, we split $L_n(D)$ into two subsequences, $\SIGMA_{a}$ and $\SIGMA_{s}$.
Formally, we set $\SIGMA_{a} = \langle X_1,X_2,\ldots,X_{\THRESHOLDALG}\rangle$, and $\SIGMA_{s} = \langle X_{\THRESHOLDALG + 1},X_{\THRESHOLDALG + 2},\ldots, X_n\rangle$.

By construction, $\POPC{\EPS}{\DELTA}$ uses $\DNF$ on the first $\THRESHOLDALG$ items while counting the number of items of each size.
After having seen $\THRESHOLDALG$ items, it creates the predicted distribution $\bm{{\hat{f}^\THRESHOLDALG}} = \PACALG(X_1,X_2,\ldots,X_{\THRESHOLDALG})$, by Lines~\ref{line:create_preds_start}-\ref{line:create_preds_end} in Algorithm~\ref{alg:popc}.
By construction of $\THRESHOLDALG$ and Proposition~\ref{prop:pac},
\begin{align*}
P\left( L^1(\bm{{\hat{f}^\THRESHOLDALG}},D) \leqslant \frac{1}{2 \cdot (m_{k,\tilde{\EPS}} + 1)}  \right) \geqslant \sqrt{1 - \DELTA}.
\end{align*}
Hence, by construction of $\bm{{\hat{f}^\THRESHOLDALG}}$ and the definition of $L^1$,
\begin{align*}
P\left( \sum_{i=1}^k \abs{\hat{f}_i^\THRESHOLDALG - p_i} \leqslant \frac{1}{2 \cdot (m_{k,\tilde{\EPS}} + 1)}  \right) \geqslant \sqrt{1-\DELTA}.
\end{align*}
Denote by $\bm{{f^{\SIGMA_{s}}}}$ the true frequencies of $\SIGMA_{s} = \langle X_{\THRESHOLDALG+1},X_{\THRESHOLDALG+2},\ldots,X_n\rangle$.
Since $n \geqslant 2 \cdot \THRESHOLDALG$, we know that $\abs{\SIGMA_{s}} \geqslant \THRESHOLDALG$, and so, by similar arguments as above, 
\begin{align*}
P\left( \sum_{i=1}^k \abs{f^{\SIGMA_{s}}_i - p_i} \leqslant \frac{1}{2 \cdot (m_{k,\tilde{\EPS}} + 1)}  \right) \geqslant \sqrt{1-\DELTA}.
\end{align*}
Let $E_{\bm{{\hat{f}^\THRESHOLDALG}}}$ be the event $\sum_{i=1}^k \abs{\hat{f}_i^\THRESHOLDALG - p_i} \leqslant \frac{1}{2 \cdot (m_{k,\tilde{\EPS}} + 1)}$, and $E_{\bm{{f^{\SIGMA_{s}}}}}$ be the event $\sum_{i=1}^k \abs{f^{\SIGMA_{s}}_i - p_i} \leqslant \frac{1}{2 \cdot (m_{k,\tilde{\EPS}} + 1)}$.
Since $E_{\bm{{\hat{f}^\THRESHOLDALG}}}$ and $E_{\bm{{f^{\SIGMA_{s}}}}}$ are independent, 
\begin{align*}
P\left( E_{\bm{{\hat{f}^\THRESHOLDALG}}} \textbf{ and }E_{\bm{{f^{\SIGMA_{s}}}}}\right) \geqslant 1-\DELTA,
\end{align*}
and so, with probability at least $1-\DELTA$, we have that
\begin{align}
L^1(\bm{{\hat{f}^\THRESHOLDALG}} , \bm{{f^{\SIGMA_{s}}}} ) =  \sum_{i=1}^k \abs{\hat{f}_i^\THRESHOLDALG - f_i^{\SIGMA_{s}}} \leqslant \sum_{i=1}^k \abs{\hat{f}_i^\THRESHOLDALG - p_i} + \sum_{i=1}^k  \abs{f_i^{\SIGMA_{s}} - p_i} < \frac{1}{m_{k,\tilde{\EPS}}}.
\label{eq:closeApprox}
\end{align}
This means that the predictions $\POPC{\EPS}{\DELTA}$ creates are very close to the true frequencies of the remainder of the instance, $\SIGMA_s$, with high probability.

Next, by construction of $\POPC{\EPS}{\DELTA}$, we get that
\begin{align*}
\POPC{\EPS}{\DELTA}(\SIGMA_n(D)) \geqslant \PC{\tilde{\EPS}}(\SIGMA_s,\bm{{\hat{f}^\THRESHOLDALG}}).
\end{align*}
Then, as long as we can verify that that
\begin{align}\label{eq:to_check_average_case}
\PC{\tilde{\EPS}}(\SIGMA_s,\bm{{\hat{f}^\THRESHOLDALG}}) \geqslant (1 - \EPS) \cdot \OPT(\SIGMA_s),
\end{align}
whenever $L^1(\bm{{\hat{f}^\THRESHOLDALG}} , \bm{{f^{\SIGMA_{s}}}} ) < \frac{1}{m_{k,\tilde{\EPS}}}$, we get that
\begin{align*}
\POPC{\EPS}{\DELTA}(\SIGMA_n(D)) &\geqslant \PC{\tilde{\EPS}}(\SIGMA_s,\bm{{\hat{f}^\THRESHOLDALG}}) \\
&\geqslant (1-\EPS) \cdot \OPT(\SIGMA_s) \\
&\geqslant (1-\EPS) \cdot \OPT(\SIGMA_n(D)) - 2 \cdot \THRESHOLDALG,
\end{align*}
Since $P(L^1(\bm{{\hat{f}^\THRESHOLDALG}} , \bm{{f^{\SIGMA_{s}}}} ) < \frac{1}{m_{k,\tilde{\EPS}}}) \geqslant 1 - \DELTA$, by Equality~\ref{eq:closeApprox}, we can write
\begin{align*}
P\left( \POPC{\EPS}{\DELTA}(\SIGMA_n(D)) \geqslant (1- \EPS) \cdot \OPT(\SIGMA_n(D)) - 2 \cdot \THRESHOLDALG \right) \geqslant 1-\DELTA,
\end{align*}
which completes the proof.

It remains to prove that Equation~\eqref{eq:to_check_average_case} holds whenever $L^1(\bm{{\hat{f}^\THRESHOLDALG}} , \bm{{f^{\SIGMA_{s}}}} ) < \frac{1}{m_{k,\tilde{\EPS}}}$.
To this end, assume that $L^1(\bm{{\hat{f}^\THRESHOLDALG}} , \bm{{f^{\SIGMA_{s}}}} ) < \frac{1}{m_{k,\tilde{\EPS}}}$.

Let $g_{\tilde{\EPS}}$ be the number of groups that $\PC{\tilde{\EPS}}$ would complete on instance $(\SIGMA_s,\bm{{f^{\SIGMA_s}}})$, that is, with perfect predictions. 
Moreover, let $\PLACEHOLDERS{\SIGMA_s,\tilde{\EPS}} = \OPT[\langle \lfloor f^{\SIGMA_s}_1\cdot m_{k,\tilde{\EPS}} \rfloor  , \ldots, \lfloor f^{\SIGMA_s}_k \cdot m_{k,\tilde{\EPS}} \rfloor  \rangle]$, and let $\PLACEHOLDERS{\THRESHOLDALG,\tilde{\EPS}} = \OPT[\langle \lfloor \hat{f}^\THRESHOLDALG_1 \cdot m_{k,\tilde{\EPS}} \rfloor , \ldots, \lfloor \hat{f}_k^\THRESHOLDALG \cdot m_{k,\tilde{\EPS}} \rfloor  \rangle]$, where items have been replaced with placeholders.

First, we compare the number of items of size $s_i$ in $\PLACEHOLDERS{\SIGMA_s,\tilde{\EPS}}$ compared to $\PLACEHOLDERS{\THRESHOLDALG,\tilde{\EPS}}$.
To this end, for all $i \in [k]$, set $\MU_i = \abs{\lfloor \hat{f}_i^\THRESHOLDALG \cdot m_{k,\tilde{\EPS}} \rfloor - \lfloor f^{\SIGMA_s}_i \cdot m_{k,\tilde{\EPS}} \rfloor}$. 
Then,
\begin{align*}
\MU_i \leqslant \abs{ \hat{f}_i^\THRESHOLDALG \cdot m_{k,\tilde{\EPS}} - f^{\SIGMA_s}_i \cdot m_{k,\tilde{\EPS}} } + 1 = \abs{\hat{f}_i^\THRESHOLDALG - f_i^{\SIGMA_s}}\cdot m_{k,\tilde{\EPS}} + 1.
\end{align*}
Since $L^1(\bm{{\hat{f}^\THRESHOLDALG}} , \bm{{f^{\SIGMA_{s}}}} ) < \frac{1}{m_{k,\tilde{\EPS}}}$, we get that $\sum_{i=1}^k \abs{\hat{f}_i^\THRESHOLDALG - f_i^{\SIGMA_s}} < \frac{1}{m_{k,\tilde{\EPS}}}$, which implies that $\abs{\hat{f}_i^\THRESHOLDALG - f_i^{\SIGMA_s}} < \frac{1}{m_{k,\tilde{\EPS}}}$, for all $i \in [k]$.
Hence, $\MU_i < 2$ for all $i \in [k]$, and since $\MU_i \in \NN$, we get that $\MU_i \in \{0,1\}$, for all $i \in [k]$.

Next, we lower bound $\PC{\tilde{\EPS}}(\SIGMA_s,\bm{{\hat{f}^\THRESHOLDALG}})$, as a function of $\PROFIT{\PLACEHOLDERS{\THRESHOLDALG,\tilde{\EPS}}}$ and $g_{\tilde{\EPS}}$.
Since $\PC{\tilde{\EPS}}$ would complete $g_{\tilde{\EPS}}$ groups on instance $(\SIGMA_s,\bm{{f^{\SIGMA_s}}})$, then, for all $i \in [k]$, $\SIGMA_s$ contains at least $g_{\tilde{\EPS}} \cdot \lfloor f^{\SIGMA_s}_i \cdot m_{k,\tilde{\EPS}} \rfloor$ items of size $s_i$.
Since $\MU_i \in \{0,1\}$ for all $i \in [k]$, then, on instance $(\SIGMA_s,\bm{{\hat{f}^\THRESHOLDALG}})$, $\PC{\tilde{\EPS}}$ fills all placeholders of size $s_i$ in $g_{\tilde{\EPS}}$ groups, except at most $g_{\tilde{\EPS}}$.
Hence,
\begin{align*}
\PC{\tilde{\EPS}}(\SIGMA_s,\bm{{\hat{f}^\THRESHOLDALG}}) \geqslant g_{\tilde{\EPS}} \cdot \PROFIT{\PLACEHOLDERS{\THRESHOLDALG,\tilde{\EPS}}} - g_{\tilde{\EPS}} \cdot k.
\end{align*}
For the rest of this proof, we use a strategy as in the proof of Theorem~\ref{thm:consistency_of_profilecovering}. 

To this end, let $N$ be the covering obtained by creating a copy of $\OPT[\SIGMA_s]$, from which we have removed a number of bins of type $t \in \mathcal{T}_S$, such that the number of bins of type $t$ is divisible by $g_{\tilde{\EPS}}$, for all $t \in \mathcal{T}_S$. 
By similar arguments as in Lemma~\ref{lem:profit_groups_perfect_predictions}, we get that $\PROFIT{N} \geqslant (1-\frac{\tilde{\EPS}}{3}) \cdot \OPT(\SIGMA_s)$. 

Next, observe that $N$ is comprised of $g_{\tilde{\EPS}}$ identical coverings $\overline{N}$.
Since $n \geqslant b$, we can write $\abs{\SIGMA_s} \geqslant m_{k,\tilde{\EPS}}^2 + m_{k,\tilde{\EPS}}$.
Hence, by a similar argument as in the proof of Lemma~\ref{lem:profit_groups_perfect_predictions}, we have $n_i^{\overline{N}} \leqslant n_i^{\PLACEHOLDERS{\SIGMA_s,\tilde{\EPS}}} + 1 \leqslant n_i^{\PLACEHOLDERS{\THRESHOLDALG,\tilde{\EPS}}} + 2$, for all $i \in [k]$, and thus
\begin{align*}
\PROFIT{\PLACEHOLDERS{\THRESHOLDALG,\tilde{\EPS}}} \geqslant \PROFIT{\overline{N}} - 2 \cdot k.
\end{align*}
Moreover, as in Lemma~\ref{lem:profit_groups_perfect_predictions}, it holds that $k \leqslant \frac{\frac{\tilde{\EPS}}{3} \cdot \PROFIT{\overline{N}}}{1 - \frac{\tilde{\EPS}}{3}}$, and so
\begin{align*}
\PROFIT{\PLACEHOLDERS{\THRESHOLDALG,\tilde{\EPS}}} \geqslant \PROFIT{\overline{N}} - 2 \cdot \frac{\frac{\tilde{\EPS}}{3} \cdot \PROFIT{\overline{N}}}{1 - \frac{\tilde{\EPS}}{3}} \geqslant (1-\tilde{\EPS}) \cdot \PROFIT{\overline{N}}.
\end{align*}
Conclusively,
\begin{align*}
\PC{\tilde{\EPS}}(\SIGMA_s,\bm{{\hat{f}^\THRESHOLDALG}}) &\geqslant g_{\tilde{\EPS}} \cdot (\PROFIT{\PLACEHOLDERS{\THRESHOLDALG,\tilde{\EPS}}} - k) \geqslant g_{\tilde{\EPS}} \cdot \left((1-\tilde{\EPS}) \cdot \PROFIT{\overline{N}} - \frac{\frac{\tilde{\EPS}}{3} \cdot \PROFIT{\overline{N}}}{1 - \frac{\tilde{\EPS}}{3}} \right) \\
&\geqslant g_{\tilde{\EPS}} \cdot \left(1-\frac{5}{3} \cdot \tilde{\EPS}\right) \cdot \PROFIT{\overline{N}} \geqslant\left(1-\frac{5}{3} \cdot \tilde{\EPS}\right) \cdot \left( 1 - \frac{\tilde{\EPS}}{3} \right)  \cdot \OPT(\SIGMA_s) \\
&= (1- 2 \cdot \tilde{\EPS}) \cdot \OPT(\SIGMA_s) = (1-\EPS) \cdot \OPT(\SIGMA_s).
\end{align*}
\end{proof}

\section{Concluding Remarks}
We studied the power of frequency predictions in improving the performance of online algorithms for the discrete bin cover problem. In particular, we showed that when input is adversarially generated, frequency predictions (from historical data) can help design algorithms with adjustable trade-offs between consistency and robustness. In particular, one can achieve near-optimal solutions, assuming predictions are error-free. On the other hand, when input is generated stochastically, we showed that frequencies could be learned from an input prefix of constant length to achieve solutions that are arbitrarily close to optimal with arbitrarily high confidence. 
An interesting variant of the problem concerns inputs generated adversarially but permuted randomly. This setting is in line with recent work on the analysis of algorithms with random order input (see, e.g.,~\cite{GKL22,BBFL23}). We expect that our algorithm for the stochastic setting can still be applied to this setting to achieve close to optimal solutions with high confidence, although a different analysis is needed.

\bibliographystyle{plain}
\bibliography{refs.bib}

\newpage

\appendix
\section*{Appendix}
\section{Missing Proofs}\label{sec:online_setting}

\begin{theorem}\label{thm:impossibility_deterministic_app}
Let $\ALG$ be a deterministic online algorithm for $\DBC{k}{}$, with $k \geqslant 5$. 
Then, $\CR{\ALG} \leqslant \frac{1}{2} + \frac{1}{H_{k-1}}$, where $H_{k-1} = \sum_{i=1}^{k-1} \frac{1}{i}$.
\end{theorem}
\begin{proof}
Suppose that $\CR{\ALG} = \frac{1}{2} + \mu$, for some online algorithm, $\ALG$, and some $\frac{1}{2} \geqslant \mu > \frac{1}{H_{k-1}}$.
Since $k \geqslant 5$, such $\mu$ exists.
Let $n \gg k$, and set
\begin{align*}
\SIGMA_i^n = \left\langle \left( \frac{1}{k} \right)^n , \left( \frac{k-i}{k} \right)^{\left\lfloor \frac{n}{i} \right\rfloor} \right\rangle,
\end{align*}
for each $i \in [k-1]$.
Clearly, $\OPT(\SIGMA^n_i) = \left\lfloor \frac{n}{i} \right\rfloor$, and since $\CR{\ALG} = \frac{1}{2} + \mu$, then
\begin{align}\label{eq:impossibility_performance_guarantee}
\ALG(\SIGMA_i^n) \geqslant \left( \frac{1}{2} + \mu \right) \cdot \left\lfloor \frac{n}{i} \right\rfloor - b,
\end{align}
for all $i \in [k-1]$.
Since $\ALG$ is deterministic, and the first $n$ requests of all $\SIGMA_i^n$'s are identical, $\ALG$ distributes the first $n$ items identically on all instances $\SIGMA_i^n$.
For each $j = 1,2,\ldots,k-1$, let $q_j$ be the number of bins in $\ALG[\SIGMA_i^n]$ that receives between $j$ and $k-1$ items of size $\frac{1}{k}$, and let $q_k$, be the number of bins that receive at least $k$ such items.
By definition, $\sum_{i=1}^{k-1} q_i \leqslant n - k \cdot q_k$. 
Then, for each $i \in [k-1]$,
\begin{align}\label{eq:impossibility_bad_instance}
\ALG(\SIGMA_i^n) \leqslant q_i + \frac{\left\lfloor \frac{n}{i} \right\rfloor - q_i}{2} + q_k = \frac{\left\lfloor \frac{n}{i} \right\rfloor + q_i}{2} + q_k.
\end{align}
Combining Equations~\eqref{eq:impossibility_performance_guarantee} and~\eqref{eq:impossibility_bad_instance} then, for all $i \in [k-1]$,
\begin{align*}
\left( \frac{1}{2} + \mu \right) \cdot \left\lfloor \frac{n}{i} \right\rfloor - b \leqslant \frac{\left\lfloor \frac{n}{i} \right\rfloor + q_i}{2} + q_k,
\end{align*}
and so
\begin{align*}
\mu \cdot \left\lfloor \frac{n}{i} \right\rfloor - b \leqslant \frac{q_i}{2} + q_k.
\end{align*}
Summing over $i \in [k-1]$, 
\begin{align*}
\mu \cdot n \cdot \left( \sum_{i=1}^{k-1} \frac{1}{i} \right) - (b+\mu)\cdot(k-1) \leqslant \frac{1}{2} \cdot \left( \sum_{i=1}^{k-1} q_i \right) + (k-1) \cdot q_k.
\end{align*}
Since $\sum_{i=1}^{k-1} q_i \leqslant n - k \cdot q_k$ and $\mu \leqslant \frac{1}{2}$, 
\begin{align*}
\mu \cdot n \cdot H_{k-1} &\leqslant \frac{n - k \cdot q_k}{2} + (k-1) \cdot (q_k + b + \frac{1}{2}) \\
&\leqslant \frac{n + k \cdot q_k + (k-1) \cdot (2\cdot b + 1)}{2} \\
&\leqslant \frac{n + k \cdot \frac{n - \sum_{i=1}^{k-1} q_i }{k} + (k-1) \cdot (2\cdot b + 1)}{2} \\
&\leqslant \frac{2 \cdot n + (k-1) \cdot (2\cdot b + 1)}{2}.
\end{align*}
Hence, $\mu \leqslant \frac{1}{H_{k-1}} + \frac{(k-1) \cdot (2\cdot b + 1)}{2 \cdot n \cdot H_{k-1}}$.
Since $\frac{1}{H_{k-1}} < \mu$, and $\lim_{n\to\infty} \frac{(k-1) \cdot (2\cdot b + 1)}{2 \cdot n \cdot H_{k-1}} = 0$, there exists an $n$ such that $\frac{1}{H_{k-1}} + \frac{(k-1) \cdot (2\cdot b + 1)}{2 \cdot n \cdot H_{k-1}} < \mu$, a contradiction. 
\end{proof}

\begin{theorem}\label{thm:impossibility_randomized_app}
Let $\ALG$ be a randomized online algorithm for $\DBC{k}{}$, with $k \geqslant 5$. 
Then, $\ECR{\ALG} \leqslant \frac{1}{2} + \frac{1}{H_{k-1}}$.
\end{theorem}
\begin{proof}
Fix $k \geqslant 5$, and let $n \geqslant \frac{k^2 \cdot (H_{k-1} + 2)}{2 \cdot H_{k-1}}$.
Define the sequences $\SIGMA_i^n$ as in the proof of Theorem~\ref{thm:impossibility_deterministic_app}, 
and let $\bm{{y^n(i)}}$ be the uniform distribution of the set $\{\SIGMA_i^n\}_{i=1}^{k-1}$.
Further, fix a randomized algorithm $\ALG$, and let $\bm{{x(j)}}$ be the distribution of deterministic algorithms, $\ALG_j$, that $\ALG$ decides between.
By the arguments in the proof of Theorem~\ref{thm:impossibility_deterministic_app}, 
\begin{align*}
\ALG_j(\SIGMA_i^n) \leqslant \frac{\left\lfloor \frac{n}{i}\right\rfloor + q_i}{2} + q_k.
\end{align*}
for all $i$ and $j$.
Hence, for any $j$, we have that
\begin{align*}
\mathbb{E}_{\bm{{y^n(i)}}}[\ALG_j(\SIGMA_i^n)] &\leqslant \frac{ \sum_{i=1}^{k-1} \left( \frac{\left\lfloor \frac{n}{i}\right\rfloor + q_i}{2} + q_k \right)}{k-1} \leqslant \frac{n \cdot H_{k-1} + n + (k-2) \cdot q_k}{2\cdot(k-1)}.
\end{align*}
Since this holds for all deterministic strategies, $\ALG_j$, we get that
\begin{align*}
\sup_j \mathbb{E}_{\bm{{y^n(i)}}}[\ALG_j(\SIGMA_i^n)] \leqslant \frac{n \cdot H_{k-1} + n + (k-2) \cdot q_k}{2\cdot(k-1)}.
\end{align*}
Further,
\begin{align*}
\mathbb{E}_{\bm{{y^n(i)}}}[\OPT(\SIGMA_i^n)] = \frac{1}{k-1} \cdot \left( \sum_{i=1}^{k-1} \left\lfloor \frac{n}{i} \right\rfloor \right) \geqslant \frac{n \cdot H_{k-1} - k + 1}{k-1}.
\end{align*}
Hence, using that $q_k \leqslant \frac{n}{k}$,
\begin{align*}
\frac{\sup_j \mathbb{E}_{\bm{{y^n(i)}}}[\ALG_j(\SIGMA_i^n)]}{\mathbb{E}_{\bm{{y^n(i)}}}[\OPT(\SIGMA_i^n)]} &\leqslant \frac{\frac{n \cdot H_{k-1} + n + (k-2) \cdot q_k}{2\cdot(k-1)}}{\frac{n \cdot H_{k-1} - k + 1}{k-1}} \leqslant \frac{n \cdot H_{k-1} + 2\cdot n - \frac{2\cdot n}{k}}{2\cdot n \cdot H_{k-1} - 2\cdot k}.
\end{align*}
Since $n \geqslant \frac{k^2 \cdot (H_{k-1} + 2)}{2 \cdot H_{k-1}}$, then $\frac{n \cdot H_{k-1} + 2\cdot n - \frac{2\cdot n}{k}}{2\cdot n \cdot H_{k-1} - 2\cdot k} \leqslant \frac{1}{2} + \frac{1}{H_{k-1}}$, and so
\begin{align*}
\liminf_{n\to\infty} \frac{\sup_j \mathbb{E}_{\bm{{y^n(i)}}}[\ALG_j(\SIGMA_i^n)]}{\mathbb{E}_{\bm{{y^n(i)}}}[\OPT(\SIGMA_i^n)]} \leqslant \frac{1}{2} + \frac{1}{H_{k-1}}.
\end{align*}
Lastly, since $\mathbb{E}_{\bm{{y^n(i)}}}[\OPT(\SIGMA_i^n)] = \left\lfloor \frac{n}{i} \right\rfloor$, then $\liminf_{n\to\infty} \mathbb{E}_{\bm{{y^n(i)}}}[\OPT(\SIGMA_i^n)] = \infty$, as so, by Yao's principle, $\ECR{\ALG} \leqslant \frac{1}{2} + \frac{1}{H_{k-1}}$. 
\end{proof}

\end{document}